\documentclass[final]{siamltex1213}

\usepackage{amssymb,amsmath}
\usepackage{hyperref}
\newtheorem{conjecture}{Conjecture of Guruswami and Zhou}

\newtheorem{claim}{Claim}

\newtheorem{example}{Example}

\newcommand{\CSP}{\mathrm{CSP}}
\newcommand{\Pol}{\mathrm{Pol}}
\newcommand{\sgn}{\mathrm{sgn\ }}

\newcommand{\alg}[1]{\mathbf{#1}}
\newcommand{\cons}[1]{\mathcal{#1}}
\newcommand{\inst}[1]{\mathcal{#1}}
\newcommand{\algor}[1]{{\texttt{#1}}}
\newcommand{\vect}[1]{\mathbf{#1}}
\newcommand{\vc}[1]{\mathbf{#1}}
\newcommand{\roundup}[1]{\left\lceil #1 \right\rceil}
\newcommand{\rounddown}[1]{\left\lfloor #1 \right\rfloor}
\newcommand{\blockround}[1]{h(#1)}
\newcommand{\tries}[1]{t(#1)}

\newcommand{\Opt}[1]{{\rm Opt}(\inst{#1})}
\newcommand{\SDPOpt}[1]{{\rm SDPOpt}(\inst{#1})}
\newcommand{\LPOpt}[1]{{\rm LPOpt}(\inst{#1})}
\newcommand{\Val}[2]{{\rm Val}(#1,\inst{#2})}

\newcommand{\sdpv}[2]{\mathbf{#1}_{#2}} 
\newcommand{\normsq}[1]{\left|\left| #1 \right|\right|^2} 
\newcommand{\norm}[1]{\left|\left| #1 \right|\right|} 
\newcommand{\dotprod}[2]{#1 #2}

\newcommand{\biggap}{u_2}
\newcommand{\smallgap}{u_1}
\newcommand{\Scopes}{\mathcal{S}}
\newcommand{\HORN}{\algor{Horn-$k$-Sat}}
\newcommand{\UG}{\algor{Unique-Games($q$)}}
\newcommand{\SAT}{{\algor{$2$-SAT}}}
\newcommand{\LIN}{\algor{E$3$-LIN($\alg{G}$)}}
\newcommand{\CUT}{\algor{CUT}}

\newcommand{\listcut}{\mathcal{L}^{(7)}}
\newcommand{\probcut}[2]{p^{(7)}(#1,#2)}
\newcommand{\estimcut}[1]{E^{(7)}_{#1}}
\newcommand{\listi}[1]{\mathcal{L}^{(7)}_{#1}}
\newcommand{\wrongcut}{w^{(7)}}
\newcommand{\listuncut}{\mathcal{L}^{(8)}}
\newcommand{\probuncut}[2]{p^{(8)}(#1,#2)}
\newcommand{\estimuncut}[1]{E^{(8)}_{#1}}
\newcommand{\cuti}[1]{z^{(8)}_i}
\newcommand{\wronguncut}{w^{(8)}}

\title{Robustly Solvable Constraint Satisfaction Problems\thanks{Parts of this work appeared in proceedings of STOC'12.}}
        
\author{
Libor Barto\thanks{Department of Algebra, Faculty of Mathematics and Physics, Charles University in Prague,
       Sokolovsk\'a 83, 18675 Praha 8, Czech Republic,
       {\tt libor.barto@gmail.com}. Research supported by the Grant Agency of the Czech Republic, grant 13-01832S;}
\and       
Marcin Kozik\thanks{Theoretical Computer Science Department,
       Faculty of Mathematics and Computer Science
       Jagiellonian University
       ul. Prof. St. Lojasiewicza 6,
       30-348 Krakow, Poland,
       {\tt Marcin.Kozik@uj.edu.pl}.
       Research partially supported by National Science Center grant no. DEC-2011/01/B/ST6/01006.
}
}

\begin{document}

\maketitle

\begin{abstract}
An algorithm for a constraint satisfaction problem is called robust if it outputs an assignment satisfying at least $(1-g(\varepsilon))$-fraction of the constraints given a $(1-\varepsilon)$-satisfiable instance, where $g(\varepsilon) \rightarrow 0$ as $\varepsilon \rightarrow 0$.
Guruswami and Zhou conjectured a characterization of constraint languages for which the corresponding constraint satisfaction problem admits an efficient robust algorithm. This paper confirms their conjecture. 
\end{abstract}



\pagestyle{myheadings}
\thispagestyle{plain}
\markboth{L. Barto, M. Kozik}{Robust approximation}

\section{Introduction}

The constraint satisfaction problem (CSP) provides a common framework for many theoretical problems in computer science 
as well as for many applications.
An instance of the CSP consists of variables and constraints imposed on them
and the goal is to find~(or decide whether it exists) an assignment of variables which is ``best'' for given constraints.
In the decision problem for CSP we want to decide if there is an assignment satisfying all the constraints.
In Max-CSP we wish to find an assignment satisfying maximal number of constraints.
In the approximation version of Max-CSP we seek an assignment which is, in some sense, close to the optimal one. 
This paper deals with a special case of approximation: robust solvability of the CSP. 
Given an instance which is almost satisfiable~(say $(1-\varepsilon)$-fraction of the constraint can be satisfied), 
the goal is to efficiently find an almost satisfying assignment~(which satisfies at least $(1-g(\varepsilon))$-fraction of the constraints, 
where the error function $g$ satisfies $\lim_{\varepsilon \rightarrow 0} g(\varepsilon) = 0$). 

Most of the computational problems connected to CSP are hard in general. 
Therefore, when developing algorithms, one usually restricts the set of allowed instances. 
Most often the instances are restricted in two ways: 
one restricts the way in which the variables are constrained~(e.g. the shape of the hypergraph of constrained variables), 
or restricts the allowed constraint relations~(defining {\em constraint language}).
In this paper we use the second approach, i.e. 
all constraint relations must come from a fixed, finite set of relations on a domain.

Robust solvability for a fixed constraint language was first studied in a paper by Zwick~\cite{Z98}. 
The motivation behind this approach was that, in certain practical situations, instances might be close to satisfiable  -- 
for example, a small fraction of constraints might have been corrupted by noise.
An algorithm that is able to satisfy, in such a case, most of the constraints could be useful.

Zwick~\cite{Z98} concentrated on Boolean CSPs. 
He designed a semidefinite programming (SDP) based algorithm which finds $(1-O(\varepsilon^{1/3}))$-satisfying assignment
for $(1-\varepsilon)$-satisfiable instances of {\SAT} and linear programming (LP) based algorithm which finds
$(1 - O(1/\log(1/\varepsilon)))$-satisfying assignment 
for $(1-\varepsilon)$-satisfiable instances of {\HORN} (the number $k$ refers to the maximum numbers of variables in a Horn constraint). 
The quantitative dependence on $\varepsilon$ was improved for {\SAT} to $(1-O(\sqrt{\varepsilon}))$ in \cite{CMM09}. 
For {\CUT}, a special case of {\SAT}, the Goemans-Williamson algorithm~\cite{GW95} also achieves $(1-O(\sqrt{\varepsilon}))$. 
The same dependence was proved more generally for $\UG$~\cite{CMM06}%
~(where $q$ refers to the size of the domain), 
which improved  $(1 - O(\sqrt[5]{\varepsilon} \log^{1/2}(1/\varepsilon)))$ obtained in~\cite{K02}. 
For {\algor{Horn-$2$-Sat}} the exponential loss can be replaced by $(1-3\varepsilon)$~\cite{KSTW00} and even  $(1-2\varepsilon)$~\cite{GZ11}. 
These bounds for {\HORN} ($k \geq 3$), {\algor{Horn-$2$-Sat}}, {\SAT}, and {\UG} are actually essentially optimal~\cite{K02, KKMO07, GZ11} 
assuming Khot's Unique Games Conjecture~\cite{K02}.

On the negative side, if the decision problem for CSP is NP-complete for algebraic reasons~(for precise definition see~\cite{BJK00,BJK05}) 
then, given a satisfiable instance, it is NP-hard to find an assignment satisfying $\alpha$-fraction of the constraints 
for some constant $\alpha < 1$ (see \cite{KSTW00} for the Boolean case and \cite{JKK09} for the general case). 
In particular, these problems cannot admit an efficient robust algorithm unless P$=$NP. 
However, this is not the only obstacle for robust algorithms. 
In \cite{H01} H\aa stad proved that for {\LIN} (linear equations over an Abelian group $\alg{G}$ where each equation contains precisely $3$ variables) 
it is NP-hard to find an assignment satisfying $(1/|G|+\varepsilon)$-fraction of the constraints given an instance which is $(1 - \varepsilon)$-satisfiable. 
Note that the trivial random algorithm achieves $1/|G|$ in expectation.

As observed in~\cite{Z98} the above results characterize robust solvability of all Boolean CSPs, because, 
by Schaefer's theorem~\cite{Sch78}, {\LIN}, {\HORN} and {\SAT} are essentially the only CSPs with tractable decision problem.
What about larger domains? 
A natural property which distinguishes {\HORN}, {\SAT}, and {\UG} from {\LIN} and ``algebraically'' NP-complete CSPs is bounded width~\cite{FV99}. 
Briefly, a CSP has bounded width if the decision problem can be solved by checking local consistency of the instance. 
These problems were characterized independently by the authors \cite{BK09b,BKbw} and Bulatov \cite{BulBW}. 
It was proved that, in some sense, the only obstacle to bounded width is {\LIN} -- 
the same problem which is difficult for robust satisfiability. 
These facts motivated Guruswami and Zhou to conjecture~\cite{GZ11} that the class of bounded width CSPs 
coincide with the class of CSPs admitting a robust satisfiability algorithm. 

Most of the recent developments in the decision version of the CSP are based on the algebraic approach introduced by
Jeavons, Cohen and Gyssens \cite{JCG97} 
and refined by Bulatov, Krokhin and Jeavons \cite{BJK00,BJK05}.
This approach was adjusted to work with robust solvability of CSPs 
in a recent paper by Dalmau and Krokhin~\cite{DK}.  
As a consequence they proved one direction of the Guruswami--Zhou conjecture ---
if a CSP is robustly solvable then it necessarily has bounded width.
They also proved the opposite direction in the special case of width $1$ CSPs, 
and classified the robust solvability with respect to the rate of growth of the error function $f$ in the Boolean case. 
Another recent paper  by Kun, O'Donnell, Tamaki, Yoshida and Zhou \cite{KOTYZ} gives an independent proof for width $1$ CSPs.

This paper confirms the Guruswami and Zhou conjecture in full generality. 
For any bounded width CSP we give a polynomial-time randomized 
algorithm for finding an assignment satisfying  $(1 - O(\log \log (1/\varepsilon)/{\log (1/\varepsilon)}))$-fraction of constraints 
in expectation 
provided there exists an $(1-\varepsilon)$-satisfying assignment~%
(the presented derandomization achieves a worse ratio).
The proof uncovers an interesting connection between the outputs of SDP (and LP) relaxations and Prague strategies 
-- a consistency notion crucial for the bounded width characterization in~\cite{BK09b,BKbw}.

\section{Preliminaries}

\subsection{CSP and robust algorithms}

We start by defining instances of the CSP.

\begin{definition} \label{def:csp}
An \emph{instance of the $\CSP$} is a triple $\inst{I} = (V, D,$  $\cons{C})$ with
$V$ a finite set of \emph{variables}, $D$ a finite \emph{domain}, and
$\cons{C}$ a finite list of \emph{constraints}, where each constraint is a pair $C = (S,R)$ with
$S$ a tuple of variables of length $k$, called the \emph{scope} of $C$, and
$R$ a $k$-ary relation on $D$ (i.e. a subset of $D^k$), called the \emph{constraint relation} of $C$.

An instance $\inst{I}$ is \emph{trivial} if all the constraint relations are empty.

An \emph{assignment} for
$\inst{I}$ 
is a mapping $F: V \rightarrow D$. We say that $F$ \emph{satisfies} a constraint $C = (S,R)$ if
$F(S) \in R$ (where $F$ is applied component-wise). The \emph{value} of $F$, $\Val{F}{I}$, is the fraction of constraints it satisfies.
The \emph{maximal value} of $\inst{I}$ is
$
\Opt{I} = \max_{F: V \rightarrow D} \Val{F}{I}.
$  
\end{definition}

\noindent We study the CSP restricted to instances that use only relations from a fixed, finite set.

\begin{definition}
  A finite set of relations $\Gamma$ on a finite set $D$ is called a \emph{constraint language} on $D$, 
  and $D$ is called the \emph{domain} of $\Gamma$.
  An \emph{instance of $\CSP(\Gamma)$} is an instance of the $\CSP$ such that all the constraint relations are from $\Gamma$.
\end{definition}

The \emph{decision problem} for $\CSP(\Gamma)$ asks whether an input instance $\inst{I}$ of $\CSP(\Gamma)$ has a \emph{solution}, 
i.e. an assignment which satisfies all the constraints. 
The \emph{Max-CSP} for $\CSP(\Gamma)$ asks to find an assignment of maximal value, 
i.e. such that $\Val{F}{I} = \Opt{I}$. 
This problem is computationally intractable for the vast majority of constraint languages motivating the study of \emph{approximation} algorithms.

\begin{definition}
Let $\Gamma$ be a constraint language and let $\alpha, \beta$ be real numbers. 
An algorithm \emph{$(\alpha, \beta)$-approximates} $\CSP(\Gamma)$, if
it outputs an assignment $F$ with $\Val{F}{I} \geq \alpha$ for every instance $\inst{I}$ of $\CSP(\Gamma)$ such that $\Opt{I} \geq \beta$.
\end{definition}

\noindent Our interest is in CSPs which can be well approximated on instances close to satisfiable. 

\begin{definition}
We say that $\CSP(\Gamma)$ is \emph{robustly solvable} if there exists an \emph{error function} $g: [0,1] \rightarrow [0,1]$ such that
$\lim_{\varepsilon \rightarrow 0} g(\varepsilon) = 0$,
and a polynomial-time algorithm which $(1 - g(\varepsilon), 1 - \varepsilon)$-approximates $\CSP(\Gamma)$ for every $\varepsilon \in [0,1]$. 
\end{definition}

\subsection{Bounded width}

Linear equations over an Abelian group are examples of CSPs which do not have bounded width. 
While the decision problem for these CSPs are tractable, 
they are not solvable  by local propagation algorithms (unlike for example {\HORN}, {\SAT}, or {\UG}).
A nice way to formalize solvability by local propagation is using the concept of a $(k,l)$-minimal instance.
To do so we require a notion of projection: the projection of a constraint $C$ to a tuple of variables $x_1,\dotsc, x_m$ is a constraint on $(x_1,\dotsc,x_m)$ with the constraint relation consisting of all 
  $(d_1,\dotsc,d_m)$'s which can be extended to a tuple from the constraint relation of $C$.

\begin{definition}
Let $k \leq l$ be positive integers. An instance $\inst{I} = (V, D, \cons{C})$ of the CSP is
\emph{$(k,l)$-minimal}, if:
\begin{itemize}
\item Every at most $l$-element tuple of distinct variables is within the scope of some
constraint in $\cons{C}$,
\item For every tuple $S$ of at most $k$ distinct%
  \footnote{Some technical problems in definitions can be caused by a variable appearing in a constraint more than once -- 
  they do not add to the complexity of the problem considered so we disregard them here.}
  variables and every pair of constraints $C_1$  and $C_2$ from $\cons{C}$ whose scopes
  contain all variables from $S$, the projections
  to $S$ are the same. This projection is denoted by $P^{\inst{I}}_S$, or $P_S$. 
\end{itemize}
A $(k,k)$-minimal instance is also called \emph{$k$-minimal}.
\end{definition}

For fixed $k,l$ there is an obvious polynomial-time algorithm for transforming an instance $\inst{I}$ of the CSP 
to a $(k,l)$-minimal instance with the same set of solutions: 
First we add new constraints~(initially allowing all the evaluations) to ensure that the first condition is satisfied 
and then we gradually remove those tuples from the constraint relations which falsify the second condition. 
We call the resulting instance \emph{the $(k,l)$-minimal instance corresponding to $\inst{I}$}. 
The definite article is justified since it is easy to see that the obtained instance is independent on the precise order of removals. 
It is clear that if the $(k,l)$-minimal instance corresponding to $\inst{I}$ is trivial 
then the original instance had no solution. We say that $\CSP(\Gamma)$ has width $(k,l)$ if the converse is always true. 

\begin{definition}
Let $k \leq l$ be positive integers and let $\Gamma$ be a constraint language. 
We say that $\CSP(\Gamma)$ has \emph{width $(k,l)$} if every instance $\inst{I}$ of $\CSP(\Gamma)$, 
whose corresponding $(k,l)$-minimal instance is nontrivial, has a solution.

We say that $\CSP(\Gamma)$ (or $\Gamma$) has \emph{bounded width} if it has width $(k,l)$ for some $k,l$. 
\end{definition} 

Different notions of width are often used in the literature, 
but they all lead to equivalent concepts of bounded width. 
We refer to \cite{FV99, LZ07, BKL08} for formal definitions and background.

\subsection{Primitive positive definitions, polymorphisms}
Primitive positive definitions are very useful in the decision version of CSP.
We say that a relation $R$ on $D$ is \emph{primitively positively definable} (or just \emph{pp-definable}) from a constraint language $\Gamma$ if there
exists a (primitive positive) formula
$$
\phi(x_1, \dots, x_k) \equiv \exists y_1, \dots, y_l\ \psi(x_1, \dots, x_k,y_1, \dots, y_l)\enspace,
$$
where $\psi$ is a conjunction of atomic formulas using relations in $\Gamma$ and the (binary) equality relation on $D$ such that
$$
(a_1, \dots, a_k) \in R \mbox{ if and only if } \phi(a_1,\dots, a_k) \mbox{ holds }.
$$

The algebraic approach to the CSP is based on a theorem by Geiger~\cite{G68} and also by Bodarchuk et al.~\cite{BKKR69} which
says that pp-definability is in the sense of Theorem~\ref{thm:galois} controlled by certain operations called polymorphisms. 
We will discuss the impact of particular polymorphisms on complexity of CSP in next sections.

\begin{definition}
An $l$-ary operation $f$ on $D$ (i.e. a mapping $f: D^l \rightarrow D$) is \emph{compatible} with a $k$-ary relation $R$, if 
$$(f(a_1^1, \dots, a_1^l), f(a_2^1, \dots, a_2^l), \dots, f(a_k^1, \dots, a_k^l)) \in R$$
whenever $(a_1^1, \dots, a_k^1)$, $(a_1^2, \dots, a_k^2)$, \dots, $(a_1^l, \dots, a_k^l) \in R$.

We say that $f$ is a \emph{polymorphism} of a constraint language $\Gamma$, if it is compatible with every relation in $\Gamma$.
The set of all polymorphisms of $\Gamma$ will be denoted by $\Pol(\Gamma)$.
\end{definition}

An $n$-ary relation $R$ is \emph{irredundant} if for every pair of different coordinates $1 \leq i < j \leq n$, 
the relation $R$ contains a tuple $(a_1, a_2, \dots, a_n) \in R$ with $a_i \neq a_j$. 
The following theorem ties the notion of pp-definitions and polymorphisms together.

\begin{theorem} \label{thm:galois} \cite{G68,BKKR69}
Let $\Gamma$ be a constraint language on $D$ and let $R$ be a nonempty relation on $D$. Then $\Pol(\Gamma) \subseteq \Pol(R)$ if and only if $R$ is pp-definable from $\Gamma$. Moreover, if $R$ is irredundant and $\Pol(\Gamma) \subseteq \Pol(R)$ then $R$ is pp-definable from $\Gamma$ without equality.
\end{theorem}

\section{The conjecture and known reductions}
The conjecture of Guruswami and Zhou~\cite{GZ11} states
\begin{conjecture}
  Let $\Gamma$ be a constraint language. The following are equivalent:
  \begin{itemize}
    \item $\CSP(\Gamma)$ has bounded width;
    \item $\CSP(\Gamma)$ is robustly solvable.
  \end{itemize}
\end{conjecture}

The upward implication in the conjecture was proved by Dalmau and Krokhin~\cite{DK}~(assuming P$\neq$NP) 
by combining the characterization of problems of bounded width~\cite{BK09b, BKbw, BulBW} with a result of H\r astad~\cite{H01}.
Their proof uses an adjustment to the algebraic approach~(developed by its authors) which is usually used for the decision version of CSP.
This paper proves the downward direction of the conjecture. 

\subsection{Primitive positive definitions}

An important observation for decision CSPs is that we do not increase the complexity (modulo log-space reductions) 
by adding a pp-definable relation to the constraint language~\cite{JCG97}.
More importantly, from the point of view of this article,
adding pp-definable relations into the constraint language does not change the property of having bounded width~\cite{LZ06,LZ07}.

A similar fact was proved in for robust solvability~\cite{DK}, 
under additional assumption that the pp-definition does not involve the equality relation. 
To state the result concisely we introduce the notation $\CSP(\Gamma) \leq_{RA} \CSP(\Gamma')$ as 
a shorthand for: for any error function $f$ with $\lim_{\varepsilon \rightarrow 0} f(\varepsilon) = 0$, if some polynomial-time algorithm
$(1-f(\varepsilon), 1-\varepsilon)$-approximates $\CSP(\Gamma')$ for every $\varepsilon \geq 0$ then there exists a polynomial-time algorithm
that $(1-O(f(\varepsilon)),1-\varepsilon)$-approximates $\CSP(\Gamma)$ for every $\varepsilon \geq 0$.

\begin{theorem}[\cite{DK}] \label{thm:robust_pp}
Let $\Gamma$ be a constraint language on $D$ and let $R$ be a relation on $D$. If $R$ is pp-definable from $\Gamma$ without equality, then
$\CSP(\Gamma \cup \{R\}) \leq_{RA} \CSP(\Gamma)$.
\end{theorem}

As the relations pp-definable in $\Gamma$ are fully determined by polymorphims of $\Gamma$,
the complexity of the decision problem as well as the property of having bounded width for $\CSP(\Gamma)$
depends only on the algebraic structure of $\Pol(\Gamma)$.
Robust solvability%
~(including the order of the error function) 
is also ``to a large extent'' controlled by $\Pol(\Gamma)$. 
Unfortunately, we have to say ``to a large extent'' because of the disturbing fact 
that Theorem~\ref{thm:robust_pp} allows only pp-definitions without equality. The general case with equality is open.

\subsection{Cores and singleton expansions} \label{sec:core}

Another important observation for both decision CSPs and robust solvability of CSPs is that we can restrict our attention to cores.

\begin{definition}
We say that a constraint language is a \emph{core}, if all its unary polymorphisms are bijections.
\end{definition}

If $\Gamma$ is a constraint language on $D$ which is not a core we can define another constraint language 
$\Gamma'$ on a smaller domain such that $\CSP(\Gamma)$ and $\CSP(\Gamma')$ behave identically 
with respect to decision, approximation and have the same width. 
Namely, if $e$ is a non-surjective unary polymorphism
of $\Gamma$ then we define $\Gamma' = \{e(R): R \in \Gamma\}$, 
where $e(R) = \{(e(a_1), \dots, e(a_n)): (a_1, \dots, a_n) \in R\}$ (see \cite{DK} for more details).

A nontrivial fact is that we can add singleton unary relations to any core language $\Gamma$ 
without significantly changing robust solvability, complexity of the decision problem or property of having bounded width for $\CSP(\Gamma)$.

\begin{theorem}[\cite{BJK05, DK}] \label{thm:adding_consts} 
Let $\Gamma$ be a core constraint language
and let $\Gamma' = \Gamma \cup \{\{a\}: a \in D\}$, then:
\begin{itemize}
  \item $\CSP(\Gamma)$ and $\CSP(\Gamma')$ are log-space equivalent,
  \item $\CSP(\Gamma') \leq_{RA} \CSP(\Gamma) \leq_{RA} \CSP(\Gamma')$, and
  \item  $\CSP(\Gamma)$ has bounded width if and only if $\CSP(\Gamma')$ does.
\end{itemize}
\end{theorem}

We refer to $\Gamma'$ from this theorem as the \emph{singleton expansion} of $\Gamma$. 
Theorem~\ref{thm:adding_consts} implies that  
the characterization conjectured by Guruswami and Zhou needs to be verified for singleton expansions of constraint languages only.
This restricts the family of constraint languages one needs to consider and we use this fact repeatedly throughout the paper.

Another consequence of Theorem \ref{thm:adding_consts}  is that whenever $\CSP(\Gamma)$ is tractable 
then there is a polynomial-time algorithm for finding a solution. 

\begin{theorem} \label{thm:search_problem} \cite{BJK05}
Let $\Gamma$ be a constraint language such that the decision problem for $\CSP(\Gamma)$ is solvable in a polynomial time. 
Then there is a polynomial-time algorithm for finding a solution of $\CSP(\Gamma)$
\end{theorem}

A proof of this theorem uses the singleton unary relations in the singleton expansion of the core of $\Gamma$ 
to recursively set values for variables and verify if such a partial evaluation extends to a solution.

\subsection{Interpretations} \label{sec:interpret}

Primitive positive definitions can be used to compare the constraint languages on the same domain. 
A stronger tool which also enables us to compare CSPs on different domains are pp-interpretations. 
\begin{definition}
  Let $\Gamma$ be a constraint languages and:
  \begin{itemize}
    \item $U$ be a pp-definable relation in $\Gamma$,
    \item $\Theta$ be an equivalence relation on $U$ pp-definable\footnote{Throughout the definition we identify $U^2$ with an appropriate power of the domain of $\Gamma$.} in $\Gamma$,
    \item $S_i$ be relations on $U$ pp-definable\footnote{Similarly we identify powers of $U$ with~(usually higher) powers of the domain of $\Gamma$.} in $\Gamma$.
  \end{itemize}
  The language $\Gamma'=\{S_i/\Theta =\{(u_1/\Theta,\dots,u_{n_i}/\Theta) : (u_1,\dots,u_{n_i})\in S_i\}\}_i$
  on the domain $U/\Theta$~(and every language isomorphic to it) is \emph{pp-interpretable} in $\Gamma$.
\end{definition}

It was proved in \cite{BJK05, LT09} (using a slightly different language) that if $\Gamma'$ is pp-interpretable in $\Gamma$ 
then the decision problem for $\CSP(\Gamma')$ is log-space reducible to the decision problem for $\CSP(\Gamma)$;
moreover if $\CSP(\Gamma)$ has bounded width then so does $\CSP(\Gamma')$~\cite{LZ06,LZ07}.

A similar theorem, in a more restrictive setting, was proved for robust solvability \cite{DK}. 
In this setting the relation $U$ needs to be unary, and the pp-definitions of $\Theta$ and $S_i$'s cannot use equality.

\subsection{The hardness result}

A proof of the hardness part of the characterization~\cite{DK} is based on a theorem by H\r astad~\cite{H01},
in which he establishes hardness for particular CSPs connected to Abelian groups.

For a finite Abelian group $\alg{G}=(G,+)$ let $\Gamma(\alg{G})$ denotes the constraint language on the domain $D = G$ 
consisting of all relations encoding linear equations over $\alg{G}$ with $3$ variables, 
that is, relations of the form $\{(x,y,z) \in G^3: ax + by + cz = d\}$ for some $d \in G, a,b,c \in \mathbb{Z}$. 
The corresponding $\CSP(\Gamma(\alg{G}))$ is denoted by \LIN.

\begin{theorem}[\cite{H01}] \label{thm:hastad}
If $\alg{G}$ is an Abelian group with $n > 1$ elements then for every $\varepsilon > 0$ there is no polynomial-time algorithm that
$( 1/n + \varepsilon,1-\varepsilon)$-approximates $\CSP(\Gamma(\alg{G}))$ unless \em{P$=$NP}.
\end{theorem}

Turning to equivalent descriptions of problems of bounded width we restrict our attention%
~(using the results of subsection~\ref{sec:core}) to singleton expansions of languages.
Combining the results of~\cite{FV99,LZ07,BK09b,BKbw,B, BulBW} we obtain

\begin{theorem} \label{thm:bw_char}
Let $\Gamma$ be a singleton expansion of a  constraint language. The following are equivalent.
\begin{itemize}
\item[(a)] There does not exist a nontrivial Abelian group $\alg{G}$ such that $\Gamma(\alg{G})$ is pp-interpretable in $\Gamma$.
\item[(b)] $\CSP(\Gamma)$ has bounded width. 
\item[(c)] $\CSP(\Gamma)$ has width $(2,3)$.
\item[(d)] $\Pol(\Gamma)$ contains a $3$-ary operation $f_1$ and a $4$-ary operation $f_2$ such that, for all $a,b \in D$,
\begin{align*}
f_1(a,a,b) &= f_1(a,b,a) = f_1(b,a,a) = \\
           & = f_2(a,a,a,b) = \dots = f_2(b,a,a,a)
\end{align*}
and
$
f_1(a,a,a) = a.
$
\end{itemize}
\end{theorem}

A refinement of condition (a) in the previous theorem is needed for robust solvability (see~\cite{DK} for more detailed discussion and references):

\begin{theorem} \label{thm:robust_hardness}
  Let $\Gamma$ be a singleton expansion of a  constraint language. 
  The following condition is equivalent to conditions from Theorem~\ref{thm:bw_char}
\begin{itemize}
  \item[(e)] There does not exits a nontrivial Abelian group $\alg{G}$ such that 
    $\Gamma(\alg{G})$ is pp-interpretable in $\Gamma$ \emph{in the first power of the domain and using pp-definitions without equality}.
\end{itemize}
\end{theorem}
Combining this fact with the results of~\cite{DK} discussed in subsection~\ref{sec:interpret} and Theorem~\ref{thm:hastad} 
the authors of~\cite{DK} obtain the hardness proof, 
i.e. the upward direction of the conjecture of Guruswami and Zhou~(unless P=NP).

\subsection{The missing implication}

The main result of this paper proves the missing implication and therefore confirms the conjecture of Guruswami and Zhou.
As discussed in subsection~\ref{sec:core} we can, without loss of generality, assume that $\Gamma$ is a singleton expansion of a constraint language. 
This is a statement of the main theorem in the paper:

\begin{theorem} \label{THM:MAIN}
If $\Gamma$ is a singleton expansion of a constraint language and the $\CSP(\Gamma)$ has bounded width then it is robustly solvable. 
More precisely there exists a randomized polynomial-time algorithm which returns an assignment satisfying, in expectation, $(1 - O(\log \log (1/\varepsilon)/{\log (1/\varepsilon)}))$-fraction of the constraints given
a $(1-\varepsilon)$-satisfiable instance.
\end{theorem}

\subsection{An overview of the proof}

Efficient approximation algorithms are often designed through linear programming (LP) relaxations and
semidefinite programming (SDP) relaxations. For instance, the robust satisfiability algorithm for {\HORN} \cite{Z98} uses LP relaxation while
the robust satisfiability algorithms for {\SAT} and {\UG} \cite{Z98, CMM09} are SDP-based. 

Robust algorithms for all CSPs of width $1$ were independently devised in \cite{DK} and \cite{KOTYZ}. 
From the CSPs mentioned previously, this result covers {\HORN}, but not {\SAT} or {\UG}. 
The approach in \cite{KOTYZ} is close to ours so let us briefly sketch the main ideas. 

For any instance $\inst{I} = (V,D,\cons{C})$ there is a canonical 0--1 integer program with the same optimal value as Max--CSP.
It has variables $\lambda_x(a)$ for every $x \in V$ and $a \in D$ and variables $\lambda_C(\vc{a})$ 
for every constraint $C=(S,R)$ and every tuple $\vc{a} \in A^r$, where $r$ is the arity  of $C$. 
The interpretation of $\lambda_x(a) = 1$ is that variable $x$ is assigned value $a$; 
the interpretation of $\lambda_C(\vc{a})=1$ is that $S$ is assigned (component-wise) tuple $\vc{a}$. 
The value to be maximized
is then equal to
\begin{equation} \label{eq:LPsum}
\frac{1}{|\cons{C}|} \sum_{C = (S,R) \in \cons{C}} \sum_{\vc{a} \in R} \lambda_C(\vc{a}).
\end{equation}
modulo the following constraints
\begin{align*}
\sum_{a\in D} \lambda_x(a) &= 1 \text{ for every } x \in V \\
\sum_{\vc{a} : a_i = a}\lambda_C(\vc{a}) &= \lambda_{x_i}(a) \text{ for every } C=((x_1,\dots,x_r),R), i\leq r\text{ and } a\in D.
\end{align*}

By relaxing the 0--1 program allowing the variables to take values in the range $[0,1]$ instead of $\{0,1\}$, we obtain the \emph{basic linear programming relaxation} for $\inst{I}$ with possibly larger value $\LPOpt{I}$ of the sum~(\ref{eq:LPsum}). 

The robust algorithm from \protect\cite{KOTYZ} works roughly as follows. (1) Run the basic LP relaxation for $\inst{I}$, (2) use the output of LP to remove some constraints so that the remaining instance $\inst{J}$ has the property that the $1$-minimal instance corresponding to $\inst{J}$ is non-trivial, (3) return a solution of $\inst{J}$. Steps (1) and (2) can be performed on any instance of the CSP. The instance $\inst{J}$ after step (2) has a solution whenever the language has width $1$, therefore we can perform step (3) using, for instance, Theorem~\protect\ref{thm:search_problem}.

Our robust algorithm for all bounded width CSPs has the same general form. 
The differences are that we use it only for instances with at most binary constraints~(a reduction is provided in the next section). 
In step (1) we use the basic SDP relaxation instead of the basic LP relaxation, 
and in step (2) we use weak Prague instances (see Section~\ref{sec:Prague}). 

\section{Our tools and reductions}
\subsection{Reduction to constraint languages with unary and binary relations}

In this section we present a reduction which allows us to prove Theorem~\ref{THM:MAIN} in an even more  restricted setting: 
for singleton expansions of constraint languages with unary and binary constraints only.
The reduction is given in the following proposition.

\begin{proposition} \label{prop:bin}
Let $\Gamma$ be a singleton expansion of a constraint language
on the domain $D$ which contains relations of maximum arity $l$ and such that $\CSP(\Gamma)$ has bounded width. 
Then there exists  $\Gamma'$ a singleton expansion of a constraint language on $D'$ 
containing only at most binary relations such that $\CSP(\Gamma')$ has bounded width and
$\CSP(\Gamma) \leq_{RA} \CSP(\Gamma')$.
\end{proposition}

\begin{proof}
First we define the constraint language $\Gamma'$ on $D' = D^l$. For every relation $R \in \Gamma$ of arity $k$ we add to $\Gamma'$ the unary relation $R'$ defined by
$$(a_1, \dots, a_l) \in R' \quad \mbox{ iff } \quad  (a_1, \dots, a_k) \in R, $$
for every $k \leq l$ we add the binary relation 
$$
E_k = \{((a_1,\dots, a_l), (b_1, \dots, b_l)): a_1 = b_k\},
$$
and for every $(a_1, \dots, a_l) \in D'$ we add the singleton unary relation $\{(a_1, \dots, a_l)\}$.
The singletons ensure that $\Gamma'$ is a singleton expansion.
The $\CSP(\Gamma')$ has bounded width which can be seen, for instance, from Theorem~\ref{thm:bw_char}: 
If $f_1,f_2$ are polymorphisms of $\Gamma$ from this theorem, 
then the corresponding operations $f_1',f_2'$ acting coordinate-wise on $D'$ satisfy the same equations 
and it is straightforward to check that $f_1', f_2'$ are polymorphisms of $\Gamma'$.

Now, let $\inst{I} = (V,D,\cons{C})$ be an instance of $\CSP(\Gamma)$ with $\Opt{I} = 1-\protect\varepsilon$. We transform $\inst{I}$ to an instance $\inst{I}'$ of $\CSP(\Gamma')$ as follows. We keep the original variables and for every constraint $C = ((x_1, \dots, x_k),R)$ in $\cons{C}$ we introduce
a new variable $x_C$ and add $k+1$ constraints
\begin{equation} \label{eq:newc}
((x_C),R'), ((x_1,x_C),E_1), ((x_2,x_C),E_2), \dots, ((x_k,x_C), E_k).
\end{equation}
If $F: V \rightarrow D$ is an assignment for $\inst{I}$ of value $1-\varepsilon$ then  the assignment $F'$ for $\inst{I}'$ defined by
\begin{align*}
  F'(x) &= (F(x), ?, \dots, ?) \quad \mbox{ for } x \in V, \\
  F'(x_C) &= (F(x_1), \dots, F(x_k), ?, \dots, ?) \quad \\
   & \mbox{ for } C = ((x_1, \dots, x_k), R)
\end{align*}
(where ? stands for an arbitrary element of $D$)
has value at least $1-\varepsilon$ since all the binary constraints in $\inst{I}'$ are satisfied and the constraint $(x_C,R')$ is satisfied whenever $F$ satisfies $C$.

We run the robust algorithm for $\CSP(\Gamma')$ to get an assignment $G'$ for $\inst{I'}$ with value at least $1 - g(\varepsilon)$,
and we define $G(x)$, $x \in V$ to be the first coordinate of $G'(x)$. Note that, for any constraint $C$ of $\inst{I}$, if $G'$ satisfies all the constraints (\ref{eq:newc}) then $G$ satisfies $C$. Therefore the value of $G$ is at least $1 - (l+1)g(\varepsilon)$.
\end{proof}

Now to prove Theorem~\ref{THM:MAIN} for $\Gamma$ -- a singleton expansion of an arbitrary constraint language, 
we produce $\Gamma'$~(from Lemma~\ref{prop:bin}) and if Theorem~\ref{THM:MAIN} holds for $\Gamma'$~(which has at most binary constraints) it does for $\Gamma$ as well.

\subsection{LP and SDP relaxations} \label{sec:LPandSDP}

The previous subsection allows us to present a simplified version of the definition of a basic SDP relaxation~\cite{R08} which is
appropriate for languages with only unary and binary constraints.

\begin{definition}
Let $\Gamma$ be a constraint language over $D$ consisting of at most binary relations and let
 $\inst{I} = (V, D, \cons{C})$ be an instance of $\CSP(\Gamma)$ with $m$ constraints. The goal for the \emph{basic SDP relaxation} of $\inst{I}$ is to find $(|V||D|)$-dimensional real vectors
$
\sdpv{x}{a}, x \in V, a \in D
$
maximizing
\begin{equation} \label{eq:sum}
\frac{1}{m} \left(
\sum_{(x,R) \in \cons{C} } \sum_{a \in R} \normsq{\sdpv{x}{a}} + 
\sum_{((x,y),R) \in \cons{C} } \sum_{(a,b) \in R} \dotprod{\sdpv{x}{a}}{\sdpv{y}{b}}
\right) 
\end{equation}
subject to 
\begin{description}
\item[(SDP1)] $\dotprod{\sdpv{x}{a}}{\sdpv{y}{b}} \geq 0 \quad$ for all $x,y \in V, a,b \in D$
\item[(SDP2)] $\dotprod{\sdpv{x}{a}}{\sdpv{x}{b}} = 0 \quad $  for all $x \in V, a, b \in D, a \neq b$, and 
\item[(SDP3)] $\sum_{a \in D} \sdpv{x}{a} = \sum_{a\in D} \sdpv{y}{a}, \ \normsq{\sum_{a \in D}\sdpv{x}{a}} = 1$ \hfill \\
         for all $x,y \in V$. 
\end{description}
\end{definition}

The dot products $\dotprod{\sdpv{x}{a}}{\sdpv{y}{b}}$ can be thought of as weights and the goal is to find vectors so that maximum weight is given to pairs (or elements) in constraint relations. It will be convenient to use the notation 
$$\sdpv{x}{A} = \sum_{a \in A} \sdpv{x}{a}$$
for a variable $x \in V$ and a subset $A \subseteq D$, so that condition (SDP3) can be written as $\sdpv{x}{D} = \sdpv{y}{D}$, $\normsq{\sdpv{x}{D}} = 1$. The contribution of one constraint to (\ref{eq:sum}) is by (SDP3) at most $1$ and it is the greater the less weight is given to pairs (or elements) outside the constraint relation.

The optimal value for the sum (\ref{eq:sum}), $\SDPOpt{I}$, is always at least $\Opt{I}$. 
There are  algorithms~(see e.g.~\protect\cite{SDP}) that output vectors with (\ref{eq:sum}) $ \geq \SDPOpt{I} - \delta$ which are polynomial in the input size and $\log(1/\delta)$.

From the output of the basic SDP relaxation we can get a valid output of the LP relaxation by defining $\lambda_x(a) = \normsq{\sdpv{x}{a}}$ and $\lambda_{(x,y)}(a,b) = \dotprod{\sdpv{x}{a}}{\sdpv{y}{b}}$ for any constraint $((x,y),R)$. In particular, $\SDPOpt{I} \leq \LPOpt{I}$.

\section{Prague instances} \label{sec:Prague}

The proof of the characterization of bounded width CSPs in \cite{BK09b} relies on a certain consistency notion called Prague strategy. 
It turned out that Prague strategies are related to outputs of basic SDP relaxations and this connection is what made our main result possible. 
The main result actually uses a stronger result, about weaker consistency notion called weak Prague instance~\cite{BKbw}.

Terms defined below are used only for certain types of instances and constraint languages.  
In our main proof we will construct, using an output of an SDP program, a weak Prague instance in a language which is different
than the language of the original instance.
Therefore, in the remainder of this section we assume that
\begin{itemize}
\item 
$\Lambda$ is a constraint language on a domain $D$, $\Lambda$ contains only binary relations,
\item 
$\inst{J} = (V,D,\cons{C}^{\inst{J}})$ is an instance of $\CSP(\Lambda)$ such that every pair of distinct variables is the scope of at most one constraint
$((x,y),P^{\inst{J}}_{x,y})$, and
  if $((x,y),P^{\inst{J}}_{x,y}) \in \cons{C}^{\inst{J}}$ then $((y,x),P^{\inst{J}}_{y,x}) \in \cons{C}^{\inst{J}}$, where $P^{\inst{J}}_{y,x} = \{(b,a): (a,b) \in P^{\inst{J}}_{x,y}\}$. 
  (Usually the instance is clear from context and then we omit the superscripts for $P_{x,y}$'s and $\cons{C}$.)
\end{itemize}

Note that under these assumptions $\inst{J}$ is $1$-minimal if and only if every variable is in the scope of some constraint and for every constraint $((x,y), P^{\inst{J}}_{x,y})$ the projection of $P^{\inst{J}}_{x,y}$ to the first coordinate is equal to $P_x^{\inst{J}}$, where $P_x^{\inst{J}}$ are the sets from the definition of $1$-minimality.

\subsection{Weak Prague instance}

First we need to define steps and patterns.

\begin{definition}
A \emph{step} (in $\inst{J}$) is a pair of variables $(x,y)$ which is the scope of a constraint in $\inst{J}$.
A \emph{pattern from $x$ to $y$} is a sequence of variables $p = (x=x_1, x_2, \dots, x_k=y)$ such that every $(x_i,x_{i+1})$, $i = 1, \dots, k-1$ is a step.

For a pattern $p = (x_1,\dots,x_k)$ we put $-p = (x_k, \dots, x_1)$.
If $p = (x_1, \dots, x_k)$, $q = (y_1, \dots, y_l)$, $x_k = y_1$ then the concatenation of $p$ and $q$ is the pattern $p+q = (x_1, x_2, \dots, x_k = y_1, y_2, \dots, y_k)$. 
For a pattern $p$ from $x$ to $x$ and a natural number $k$, $kp$ denotes the $k$-time concatenation of $p$ with itself.
\end{definition}

Observe that from the assumptions about $\inst{J}$ it follows that $-p$ is a pattern whenever $p$ is.

\begin{definition}
  Let $p = (x=x_1,x_2, \dots, x_k=y)$ be a pattern from $x$ to $y$ in $\inst{J}$.
  A \emph{realization of $p$} is a sequence $(a_1, \dots, a_k) \in D^k$ such that $(a_i,a_{i+1}) \in P_{x_i,x_{i+1}}$ for every $1 \leq i \leq k-1$.
  
  For a subset $A \subseteq D$ we define $A+p$ as the set of the last elements of those realizations of $p$ whose first element is in $A$, that is,
  $$
  A + p = \{b \in D: (\exists \ a_1, \dots, a_{k-1} \in D) \ a_1 \in A \mbox{ and } (a_1,\dots, a_{k-1},b) \mbox{ is a realization of $p$}\}.
  $$
  Finally, we define $A - p = A + (-p)$.
\end{definition}

The addition of patterns is associative, i.e. $(A+p) + q = A + (p+q)$. 
Also note that in a $1$-minimal instance we have $A \subseteq A + p - p$ for any $A\subseteq P_x$ and any pattern $p$ from $x$.

A weak Prague instance is a $1$-minimal instance with additional requirements concerning addition of patterns.

\begin{definition} \label{def:prague}
$\inst{J}$ is a \emph{weak Prague instance} if
\begin{itemize}
\item[(P1)] $\inst{J}$ is $1$-minimal,
\item[(P2)] for every $A \subseteq P_x^{\inst{J}}$ and every pattern $p$ from $x$ to $x$, if 
$A + p = A$ then $A - p = A$, and
\item[(P3)] for any patterns $p_1,p_2$ from $x$ to $x$ and every $A \subseteq P_x^{\inst{J}}$, if
$A + p_1 + p_2 = A$ then $A + p_1 = A$.
\end{itemize}
\end{definition}

\noindent
To clarify the definition let us assume that $\inst{J}$ is $1$-minimal and consider the following digraph: vertices are all the pairs $(A,x)$, where $x \in V$ and $A \subseteq P_x^{\inst{J}}$, and $((A,x),(B,y))$ forms an edge iff $(x,y)$ is a step and $A + (x,y) = B$. 
Condition (P3) means that no strong component contains $(A,x)$ and $(A',x)$ with $A \neq A'$,
condition (P2) is equivalent (by the following lemma) to the fact that every strong component contains only undirected edges (that is, if $((A,x),(B,y))$ is an edge then so is $((B,y),(A,x))$).

\begin{lemma} \label{lem:P2star}
Let $\inst{J}$ be a $1$-minimal instance. Then (P2) is equivalent to the following condition.
\begin{itemize}
\item[(P2*)]
For every step $(x,y)$, every $A \subseteq P_x$ and every pattern $p$ from $y$ to $x$, if $A + (x,y) + p = A$ then $A + (x,y,x) = A$.
\end{itemize}
\end{lemma}
\begin{proof}
(P2*) $\Rightarrow$ (P2). If $p = (x=x_1, x_2, \dots, x_k=x)$ is a pattern from $x$ to $x$ such that $A+p = A$, then repeated application of (P2*) gives us
\begin{align*}
A + &p - p = \\
&= [A + (x_1, x_2, \dots, x_{k-1})] + (x_{k-1}, x_{k}, x_{k-1})  \\
&  \quad  + (x_{k-1}, x_{k-2}, \dots, x_1) = \\
&= A + (x_1, x_2, \dots, x_{k-1}) + (x_{k-1}, x_{k-2}, \dots, x_1) = \\
&= [A + (x_1, x_2, \dots, x_{k-2})] + (x_{k-2}, x_{k-1}, x_{k-2})  \\
&  \quad + (x_{k-2}, x_{k-3}, \dots, x_1) = \\
&= A + (x_1, x_2, \dots, x_{k-2}) + (x_{k-2}, x_{k-3}, \dots x_1) = \\
&= \dots = \\
&= A,
\end{align*}
where the second equality uses (P2*) for the set $A + (x_1, x_2,$ $ \dots, x_{k-1})$. The assumption of (P2*) is provided by a cyclic shift of the 
pattern $p$: $[A + (x_1, \dots, x_{k-1})] +$ $(x_{k-1},x_{k}) + (x_{1},\dots, x_k,\dots,x_{k-1}) = [A + (x_1, x_2, \dots, x_{k-1})]$ as $A + (x_1, \dots, x_{k-1}) +$ $(x_{k-1},x_k) =A$.
The fourth equality uses (P2*) for the set $A + (x_1, \dots, x_{k-2})$ and so on.

(P2) $\Rightarrow$ (P2*). By applying (P2) to the pattern $(x,y)+p$ we get $A + (x,y) + p - p + (y,x) = A$. From $1$-minimality it follows that
$A + (x,y) \subseteq A + (x,y) + p - p$, hence $A + (x,y,x) = (A + (x,y)) + (y,x) \subseteq (A + (x,y)+ p - p) + (y,x) = A$. The other inclusion follows again from $1$-minimality.  
\end{proof}
 
\begin{example}
  An example of a weak Prague instance, which is not a Prague strategy~\cite{BKbw} i.e. witnessing that the new notion is weaker, is
  $V = \{x,y,z\}$, $D = \{0,1\}$, $P_{x,y} = P_{x,z} = \{(0,0),(1,1)\}$, $P_{y,z} = \{(0,0),(0,1),(1,0),(1,1)\}$.

  If we change $P_{y,z}$ to $\{(0,1),(1,0)\}$ the conditions (P1) and (P2) hold but 
  $\{0\} + (x,y,z,x) + (x,y,z,x) = \{0\}$ and $\{0\} + (x,y,z,x) = \{1\}$.

  If, on the other hand, we set $P_{y,z}$ to $\{(0,0),(1,0),(1,1)\}$ then (P1) and (P3) hold while
  $\{0\} + (x,y,z,x) = \{0\}$, but $\{0\} - (x,y,z,x) = \{0,1\}$.
\end{example}

The main result of this paper relies on the following theorem.

\begin{theorem}[\cite{BKbw}] \label{thm:bw} 
  If $\CSP(\Lambda)$ has bounded width and $\inst{J}$ is a nontrivial weak Prague instance of $\CSP(\Lambda)$ then $\inst{J}$ has a solution.
\end{theorem}

\subsection{SDP and Prague instances}

We now show that one can naturally associate a weak Prague instance to an output of the basic SDP relaxation. 
This material will not be used in what follows, it is included to provide some intuition for the proof of the main theorem.

Let $\sdpv{x}{a}$, $x \in V$, $a \in D$ be arbitrary vectors satisfying (SDP1), (SDP2) and (SDP3).
We define a CSP instance $\inst{J}$ by
\begin{align*}
\inst{J} &= (V,D,\{((x,y),P_{x,y}): x,y \in V, x \neq y\}), \\
P_{x,y} &= \{(a,b): \dotprod{\sdpv{x}{a}}{\sdpv{y}{b}} > 0\},
\end{align*}
and we show that it is a weak Prague instance. 

The instance is $1$-minimal with $P_x^{\inst{J}} = \{a \in D: \sdpv{x}{a} \neq \vect{0}\}$. To prove this it is enough to verify that the projection of $P_{x,y}$ to the first coordinate is equal to $P_x^{\inst{J}}$. If $(a,b) \in P_{x,y}$, then clearly $\sdpv{x}{a}$ cannot be the zero vector, therefore $a \in P_x^{\inst{J}}$. On the other hand, if $a \in P_x^{\inst{J}}$ then 
$0 < \normsq{\sdpv{x}{a}} = \dotprod{\sdpv{x}{a}}{\sdpv{x}{D}} = \dotprod{\sdpv{x}{a}}{\sdpv{y}{D}}$
and thus at least one of the dot products $\dotprod{\sdpv{x}{a}}{\sdpv{y}{b}}$, $b \in D$ is nonzero and $(a,b) \in P_{x,y}$.

To check (P2) and (P3) we note that, for any $x,y \in V, x \neq y$ and $A \subseteq P^{\inst{J}}_x$, 
the vector $\sdpv{y}{A+(x,y)}$ has either a strictly greater length than $\sdpv{x}{A}$, 
or $\sdpv{x}{A} = \sdpv{y}{A+(x,y)}$, and the latter happens iff $A+(x,y,x) = A$ (see the proof of Claim \ref{cl:walk},
in fact, one can check that $\sdpv{y}{A + (x,y)}$ is obtained by adding to $\sdpv{x}{A}$ an orthogonal vector 
whose size is strictly greater than zero iff $A+(x,y,x) \neq A$).  
By induction, for any pattern $p$ from $x$ to $y$, the vector $\sdpv{y}{A+p}$ is either strictly longer than $\sdpv{x}{A}$, or 
$\sdpv{x}{A} = \sdpv{y}{A+p}$ and $A + p - p = A$. Now (P2) follows immediately and (P3) is also easily seen:
If $A+p+q = A$ then necessarily $\sdpv{x}{A} = \sdpv{x}{A+p}$ which is possible only if $A = A+p$.

We end this section with several remarks.
\subsubsection{Considering only the squares of length of vectors is equivalent to LP}
To prove property (P2) we only need to consider the lengths of the vectors. 
In fact, this property will be satisfied when we start with the basic linear programming relaxation~%
(and define the instance $\inst{J}$ in a similar way --- compare the end of section~\ref{sec:LPandSDP}). 
This is not the case for property (P3).

\subsubsection{This is a Prague strategy}
The above weak Prague instance is in fact a Prague strategy in the sense of \cite{BK09b}. This means that every pair of variables is the scope of a (unique) constraint and all strong components of the digraph introduced after Definition~\ref{def:prague} are complete graphs.

\subsubsection{The SDP relaxation does not guarantee a $(2,3)$-minimal instance}
There were attempts to show that the instance $\inst{J}$ is $(2,3)$-minimal after adding appropriate ternary constraints. This is equivalent to the requirement that $P_{x,y}$ is a subset of the composition of the relations
$P_{x,z}$ and $P_{z,y}$ for every $x,y,z$. The following example shows that it is not the case.
Consider $V = \{x,y,z\}$, $D = \{0,1\}$ and vectors
$\sdpv{x}{0} = (1/2, 1/2, 0)$, $\sdpv{x}{1} = (1/2,-1/2,0)$,
$\sdpv{y}{0} = (1/4, -1/4,$  $\sqrt{2}/4)$, $\sdpv{y}{1} = (3/4,1/4, -\sqrt{2}/4)$,
$\sdpv{z}{0} = (1/4, 1/4, \sqrt{2}/4)$, $\sdpv{z}{1} = (3/4,$ $ -1/4, -\sqrt{2}/4)$.
The constraint relations are then
$P_{x,y} = \{(0,1), (1,0), (1,1)\} = P_{y,x}$,
$P_{x,z} = \{(0,0), $ $(0,1), (1,$ $1)\}   = P_{z,x}^{-1}$,
$P_{y,z} = \{(0,0), (0,1), (1,0),(1,1)\}$ $ = P_{z,y}.$
The pair $(0,0) \in P_{y,z}$ is not in the composition of the relations $P_{y,x}$ and $P_{x,z}$ since there is no $a \in \{0,1\}$ such that $(0,a) \in P_{y,x}$ and $(a,0) \in P_{x,z}$.

\subsubsection{$\SDPOpt{I}=1$ implies solution}
Finally, we note that if $\inst{I}$ is an instance of the CSP with $\SDPOpt{I} = 1$ and we define $\inst{J}$ using vectors with the sum (\ref{eq:sum}) equal to 1, then a solution of $\inst{J}$ is necessarily a solution to $\inst{I}$. Showing that ``$\SDPOpt{I}=1$'' implies ``$\inst{I}$ has a solution'' was suggested as a first step to prove the Guruswami-Zhou conjecture. It indeed proved to be the right direction.

\subsection{Algebraic closure of a weak Prague instance}

The proof of correctness of the robust algorithm for bounded width CSPs obtains a solution from a  certain weak Prague instance $\inst{J}$. 
Instance $\inst{J}$ is obtained from the result of an SDP algorithm on the basic SDP relaxation of the original instance.
Unfortunately  the constraints in $\inst{J}$ does not necessarily have bounded width so we cannot directly apply Theorem {\ref{thm:bw}}.
This technical difficulty is overcome using Proposition {\ref{prop:closure}} below. Note that the solution (given by Theorem {\ref{thm:bw}}) 
to the instance given by Proposition {\ref{prop:closure}} can be outside the Prague instance $\inst{J}$.

The following lemma from~\cite{BKbw} shows a basic property of weak Prague instances. 

\begin{lemma}\label{lem:cl}
Let $\inst{J}$ be a weak Prague instance, $x\in V$, $A \subseteq P_x$, and let $p$ be a pattern from $x$ to $x$. 
Then there exists a natural number $l$ such that $A + lp + l'p = A + lp$ for every integer $l'$ and, moreover, $A \subseteq A + lp$.
\end{lemma}

The set $A+lp$ from Lemma~\ref{lem:cl} is denoted by $[A]_p$. For a singleton $A = \{a\}$ we write $[a]_p$.
We have $[A]_p + l'p = [A]_p$ for every integer $l'$ and, moreover,  $A \subseteq [A]_p$.

\begin{proposition} \label{prop:closure}
Let $\inst{J} = (V,D,\{P_{x,y}: (x,y) \in \Scopes\})$ be a weak Prague instance and let $\mathcal{F}$ be a set of operations on $D$. Then
$\inst{J}' = (V,D,\{P'_{x,y}: (x,y) \in \Scopes\})$,
where
\begin{align*}
P'_{x,y} &= \{ (f(a_1,a_2, \dots),f(b_1, b_2, \dots)): f \in \mathcal{F}, \\
 & \quad \quad (a_1,b_1), (a_2,b_2), \dots \in P_{x,y}\},
\end{align*}
is a weak Prague instance.
\end{proposition}

\begin{proof}
It is apparent that $\inst{J}'$ is $1$-minimal with
$$
P^{\inst{J}'}_{x} = P'_{x} := \{ f(a_1,a_2, \dots): f \in \mathcal{F}, \ a_1, a_2, \dots \in P_x \}.
$$
In what follows, by $A +' p$ we mean the addition computed in the instance $\inst{J}'$ while $A + p$ is computed in $\inst{J}$.
Moreover, by $f(A_1, \dots, A_k)$ we mean the set 
$ 
\{f(a_1, \dots, a_k):$ $a_1 \in A_1, a_2 \in A_2, \dots, a_k \in A_k\}
$.
 
Before proving (P2) and (P3) we make a simple observation. 
\begin{claim} \label{cl:noname}
If  $f \in \mathcal{F}$ is an operation of arity $k$, $x \in V$, $p$ is a pattern from $x$, and $A_1, \dots, A_k \subseteq P_x$,  $B \subseteq P_x'$ are such that $f(A_1, A_2, \dots, A_k) \subseteq B$,
then%
$f(A_1 + p, A_2 + p, \dots A_k + p) \subseteq B +' p$.
\end{claim}
\begin{proof}
It is enough to prove the claim for a single step $p = (x,y)$. The rest follows by induction. 
If $b \in f(A_1 + (x,y), \dots, A_k + (x,y))$ then there exist elements $b_1 \in A_1 + (x,y)$, \dots, $b_k \in A_k + (x,y)$ so that
$f(b_1, b_2, \dots, b_k) = b$. As $b_i \in A_i + (x,y)$ there are elements $a_i \in A_i$ such that $(a_i,b_i) \in P_{x,y}$ for all $1 \leq i \leq k$. But then $(f(a_1,a_2, \dots, a_l),$ $f(b_1,b_2, \dots, b_k))$ is in $P'_{x,y}$ and $f(a_1, a_2, \dots, a_k) \in f(A_1, A_2,$ $\dots, A_k) \subseteq B$, therefore $b = f(b_1,b_2, \dots, b_k) \in B +' (x,y)$.
\end{proof}

Instead of  (P2) for the instance $\inst{J}'$ we prove (P2*) from Lemma~\ref{lem:P2star}.
Let $(x,y)$ be a step, $A \subseteq P_x'$, let $p$ be a pattern from $y$ to $x$ such that $A +' (x,y) +' p = A$, and
let $a$ be an arbitrary element of $A +' (x,y,x)$. As $A +' (x,y,x) = (A +' (x,y)) +' (y,x)$, there exist $b \in A +' (x,y)$ such that $(a,b) \in P'_{x,y}$. By definition of $P'_{x,y}$, we can find $f \in \mathcal{F}$ (say, of arity $k$), elements $a_1, a_2, \dots, a_k$ in $P_x$, and
$b_1, \dots, b_k$ in $P_y$ so that $(f(a_1, a_2, \dots, a_k),f(b_1, b_2, \dots, b_k)) = (a,b)$ and $(a_i,b_i) \in P_{x,y}$ for all $1 \leq i \leq k$.

We consider the sets $[b_1]_q, [b_2]_q, \dots, [b_2]_q$  for the pattern $q = p + (x,y)$. We take $l$ to be the maximum of the numbers for $b_1, \dots, b_k$ from Lemma~\ref{lem:cl}, so $[b_i]_q = \{b_i\} + lq$. 
We get
\begin{align*}
a_i &\in \{b_i\} + (y,x)
\subseteq [b_i]_q + (y,x) 
= \\
& = [b_i]_q + p + (x,y) + (y,x) =
[b_i]_q + p,
\end{align*}
where the first step follows from $(a_i,b_i) \in P_{x,y}$, the inclusion and the first equality from Lemma~\ref{lem:cl},
and the second equality from (P2*) for the instance $\inst{J}$ (as $([b_i]_q + p) + (x,y) + p = [b_i]_q + p$).
Thus $a = f(a_1, a_2, \dots, a_{k})$ is an element of 
\begin{align*}
f([b_1]_q + p, &[b_2]_q + p, \dots, [b_k]_q + p) = \\
&=
f(\{b_1\} + lq + p, \dots, \{b_k\} + lq + p)
\end{align*} 
and this set is contained in $(A +'(x,y)) +' lq +'p = A +' (x,y) +' l(p + (x,y)) +' p = A$ by Claim~\ref{cl:noname} applied with $A_i = \{b_i\}$ and the pattern $lq + p$. We have shown that every element $a$ of $A + ' (x,y,x)$ lies in $A$. The other inclusion follows from $1$-minimality.

To prove (P3) let $x \in V$, $A \subseteq P'_x$ and let $p,q$ be patterns such that $A +' p +' q = A$. We first show that $A \subseteq A +' p$. Let $a \in P'_x$, take  $f \in \mathcal{F}$, $a_1, a_2, \dots, a_k \in P_x$ such that $f(a_1, \dots, a_k) = a$, and find $l$ so that $[a_i]_{p+q} = a_i + l(p+q)$. From (P3) for $\inst{J}$ and Lemma~\ref{lem:cl} it follows that $[a_i]_{p+q}+p = [a_i]_{p+q}$. 
By Claim~\ref{cl:noname},$a \in f([a_1]_{p+q},[a_2]_{p+q}, \dots, [a_k]_{p+q}) = f([a_1]_{p+q}+p,[a_2]_{p+q}+p, \dots, [a_k]_{p+q}+p)  \subseteq A +' l(p+q) +' p = A +' p$.
The same argument used for $A + 'p$ instead of $A$ and the patterns $q+p,q$ instead of $p+q,p$ proves $A +' p \subseteq A + 'p +' q = A$.
\end{proof}

\section{Robust algorithm for bounded width CSPs}

The final, and most technical, version of our main result follows.
Theorem~\ref{THM:MAIN} is a consequence of the following fact:

\begin{theorem} \label{THM:CORE}
Let $\Gamma$ be a core constraint language over $D$ containing at most binary relations. 
If $\CSP(\Gamma)$ has bounded width, then there exists
a randomized algorithm which given an instance $\inst{I}$ of $\CSP(\Gamma)$ and an output of the basic SDP relaxation with value at least $1 - 1/n^{4n}$ (where $n$ is a natural number) produces an assignment with value at least $1 - K/n$, where $K$ is a constant depending on $|D|$. The running time is polynomial in $m$ (the number of constraints) and $n^{n}$.
\end{theorem}

\subsection{Proof of Theorem~\ref{THM:MAIN} using Theorem~\ref{THM:CORE}}

To prove Theorem~\ref{THM:MAIN} we start with
$\Gamma$ which is the singleton expansion of a constraint language of bounded width.
By Proposition~\ref{prop:bin} we can assume that $\Gamma$ contains only at most binary relations. 

Let $\inst{I}$ be an instance of $\CSP(\Gamma)$ with $m$ constraints and let $1-\varepsilon = \Opt{I}$ where $\varepsilon$ is sufficiently small and $m$ sufficiently large. 
We need to show how to effectively find an assignment satisfying, in expectation, the promised fraction of constraints.

We first check whether $\inst{I}$ has a solution. This can be done in polynomial time since $\CSP(\Gamma)$ has bounded width. If a solution exists we  can find it in polynomial time by Theorem~\ref{thm:search_problem}. 

In the other case we know that $\varepsilon \geq 1/m$. 
We run the SDP relaxation with precision $\delta = 1/m$ and obtain vectors with the sum (\ref{eq:sum}) equal to $v \geq \SDPOpt{I}-1/m$. 
Finally, we execute the algorithm provided in Theorem~\ref{THM:CORE} with the following choice of $n$.
\begin{equation*}
n = \rounddown{\frac{\log \omega}{4 \log \log \omega}}, \quad \mbox{ where } \omega 
  = \min\left\{\frac1{1-v},m \right\}.
\end{equation*}
The assumption is satisfied,  because $v \geq 1 - 1/n^{4n}$ is equivalent to $n^{4n} \leq 1/(1-v)$ and 
\begin{align*}
 n^{4n} &= 2^{4n \log n} \leq 2^{4 \frac{\log \omega}{4 \log \log \omega} \log \frac{\log \omega}{4\log \log \omega}} < \\
 &< 2^{\frac{\log \omega}{\log \log \omega} \log \log \omega} =  \omega \leq 1/(1-v).
\end{align*}
The algorithm runs in time polynomial in $m$ as $n^{n} < n^{4n} \leq \omega \leq m$.
To estimate the fraction of satisfied constraints, observe that $v \geq \Opt{I} - 1/m = 1 - \varepsilon - 1/m \geq 1 - 2\varepsilon$, so $1/(1-v) \geq 1/(2\varepsilon)$, and also $m \geq 1/\varepsilon$, therefore $\omega \geq 1/(2\varepsilon)$.
The fraction of satisfied constraints is, in expectation, at least $1 - K/n$ and
\begin{align*}
\frac{n}{K} &\geq \frac1K \left( \frac{\log \omega}{4\log \log \omega} - 1 \right) \geq \\
&\geq K_3 \frac{\log (1/(2\varepsilon)) 
\geq K_4 \frac{\log (1/\varepsilon)}{\log \log (1/\varepsilon)},}{R: strange linebreaks in the second and third displayed lines.}
\end{align*}
where $K_3, K_4$ are suitable constants. Therefore the fraction of
satisfied constraints is at least
$$
1 - O\left(\frac{\log \log (1/\varepsilon)}{\log (1/\varepsilon)}\right).
$$
\subsection{Proof of Theorem~\ref{THM:CORE}}

Let $\inst{I} = (V, D, \cons{C})$ be an instance of $\CSP(\Gamma)$ with $m$ constraints and let $\sdpv{x}{a}$, $x \in V$, $a \in D$ be vectors satisfying (SDP1), (SDP2), (SDP3) such that the sum (\ref{eq:sum}) is at least $1 - 1/n^{4n}$.
Without loss of generality we assume that $n > |D|$.

Let us first briefly sketch the idea of the algorithm. 
The aim is to define an instance $\inst{J}$ in a similar way as in the previous section ($\inst{J}$ is defined after Claim~\ref{cl:works}), but instead of all pairs with nonzero weight we only include pairs of  weight greater than a threshold (chosen in Step 1). This guarantees that every solution to $\inst{J}$ satisfies all the constraints of $\inst{I}$ which do not have large weight on pairs outside the constraint relation (the bad constraints are removed in Step 3). The instance $\inst{J}$ (more precisely, its algebraic closure) has a solution by Theorem~\ref{thm:bw} as soon as we ensure that it is a weak Prague instance. Property (P1) is dealt with in a similar way as in \cite{KOTYZ}: We keep only constraints with a gap -- all pairs have either smaller weight than the threshold, or significantly larger (Step 2). This also gives a property similar to the one in the motivating discussion in the previous section: The vector $\sdpv{y}{A+(x,y)}$ is either significantly longer than $\sdpv{x}{A}$ or these vectors are almost the same. However, large amount of small differences can add up, so we need to continue taming the instance. In Steps 4 and 5 we divide the unit ball into layers and remove some constraints so that almost the same vectors of the form $\sdpv{x}{A}$, $\sdpv{y}{A+(x,y)}$ never lie in different layers. This already guarantees property (P2). For property (P3) we use ``cutting by hyperplanes'' idea from \cite{GW95}. We choose sufficiently many hyperplanes so that every pair
$\sdpv{x}{A}$, $\sdpv{x}{B}$ of different vectors in the same layer is cut (the bad variables are removed in Step 7) and we do not allow almost the same vectors for different variables to cross the hyperplane (Step~8).  

The description of the algorithm follows.
\begin{itemize}
\item[1.] Choose $r \in \{1, 2, \dots, n-1\}$ uniformly at random.
\item[2.] Remove from $\cons{C}$ all the unary constraints $(x,R)$ such that $\normsq{\sdpv{x}{a}} \in [n^{-4r-4},n^{-4r})$ for some $a \in D$
and all the binary constraints $((x,y),R)$ such that $\dotprod{\sdpv{x}{a}}{\sdpv{y}{b}} \in [n^{-4r-4},n^{-4r})$ for some $a,b \in D$.
\item[3.] Remove from $\cons{C}$ all the unary constraints $(x,R)$ such that $\normsq{\sdpv{x}{a}} \geq n^{-4r}$ for some $a \not\in R$
and all the binary constraints $((x,y),R)$ such that $\dotprod{\sdpv{x}{a}}{\sdpv{y}{b}} \geq n^{-4r}$ for some $(a,b) \not\in R$.
\end{itemize}
Let 
$$\smallgap = 2|D|^2n^{-4r-4} \ \mbox{ and } \ \biggap = n^{-4r} - \smallgap.$$ 
For two real numbers $\gamma, \psi \neq 0$ we denote by $\gamma \div \psi$ the greatest integer $i$ such that $\gamma - i \psi > 0$ and this difference is denoted by $\gamma \mod \psi$.

\begin{itemize}
\item[4.] Choose $s \in [0, \biggap]$ uniformly at random.
\item[5.] Remove from $\cons{C}$ all the binary constraints $((x,y),R)$ such that 
$|\normsq{\sdpv{x}{A}} - \normsq{\sdpv{y}{B}}| \leq \smallgap$
and $(\normsq{\sdpv{x}{A}} - s) \div \biggap \neq (\normsq{\sdpv{y}{B}} - s) \div \biggap$ for some $A,B \subseteq D$. 
\end{itemize}
The remaining part of the algorithm uses the following definitions.  
For all $x \in V$ let 
$$P_x = \{ a \in D: \normsq{\sdpv{x}{a}} \geq n^{-4r} \}.$$
For a vector $\vect{w}$ we put 
$$\blockround{\vect{w}} = (\normsq{\vect{w}} - s) \div \biggap$$ 
and
$$\tries{\vect{w}} = \roundup{\pi(\log n) n^{2r} \min\{\sqrt{(\blockround{\vect{w}}+2) \biggap}, 1\}}.$$
We say that $\vect{w}_1$ and $\vect{w}_2$ are 
\emph{almost the same} if 
$\blockround{\vect{w}_1} = \blockround{\vect{w}_2}$ and $\normsq{\vect{w}_1 - \vect{w}_2} \leq \smallgap$. 

\begin{itemize}
  \setlength{\itemsep}{1pt}
  \setlength{\parskip}{0pt}
  \setlength{\parsep}{0pt}
\item[6.] Choose unit vectors $\vect{q}_1$, $\vect{q}_2$, \dots, $\vect{q}_{\roundup{\pi(\log n)n^{2n}}}$ independently and uniformly at random.
\item[7.] We say that a variable $x \in V$ is \emph{uncut} if there exists $A,B \subseteq P_x$, $A \neq B$ such that $\blockround{\sdpv{x}{A}}=\blockround{\sdpv{x}{B}}$ and $\sgn \dotprod{\sdpv{x}{A}}{\vect{q}_i} = \sgn \dotprod{\sdpv{x}{B}}{\vect{q}_i}$ for every 
$1 \leq i \leq \tries{\sdpv{x}{A}}$ (in words, no hyperplane determined by the first $\tries{\sdpv{x}{A}} = \tries{\sdpv{x}{B}}$ vectors $\vect{q}_i$ cuts the vectors $\sdpv{x}{A}$, $\sdpv{x}{B}$).
Remove from $\cons{C}$ all the constraints whose scope contains an uncut variable.
\item[8.] Remove from $\cons{C}$ all the binary constraints $((x,y),R)$ for which there exist $A \subseteq P_x, B \subseteq P_y$ such that
 $\sdpv{x}{A}$, $\sdpv{y}{B}$ are almost the same and $\sgn \dotprod{\sdpv{x}{A}}{\vect{q}_i} \neq \sgn \dotprod{\sdpv{y}{B}}{\vect{q}_i}$ for some $1 \leq i \leq \tries{\sdpv{x}{A}}$.
\item[9.] Use the remaining constraints to construct a weak Prague instance, close it under polymorphisms~(comp. Proposition~{\ref{prop:closure}}) and compute a solution.

\end{itemize}

\begin{claim} \label{cl:works}
Expected fraction of constraints removed in steps $2$, $3$, $5$, $7$ and $8$ is at most $K / n$ for some constant $K$. 
\end{claim}

\textbf{Remark. }
The constant $K$ depends exponentially on the size of the domain $|D|$.

\begin{proof}

\textbf{Step 2.}
For each binary constraint there are $|D|^2$ choices for $a,b \in D$ and therefore at most $|D|^2$ bad choices for $r$. For a unary constraint the number of bad choices is at most $|D|$. Thus the probability that a given constraint will be removed is at most $|D|^2/(n-1)$ and it follows that the expected fraction of removed constraints is at most $|D|^2/(n-1)$.

\textbf{Step 3.}
The contribution of every removed constraint to the sum (\ref{eq:sum}) is at most $1 - n^{-4r} \leq 1 - n^{-4n+4}$. If more than $\gamma$-fraction of the constraints
is removed than the sum is at most $1/m( (1-\gamma)m + \gamma m(1 - n^{-4n+4})) = 1-\gamma n^{-4n+4}$. Since (\ref{eq:sum}) $\geq 1 - 1/n^{4n}$, we have $\gamma \leq 1/n^4$.

\textbf{Step 5.} 
For every constraint $((x,y),R)$ and every $A,B \subseteq D$ such that $|\normsq{\sdpv{x}{A}} - \normsq{\sdpv{y}{B}}| \leq \smallgap$, $\norm{\sdpv{x}{A}} \leq \norm{\sdpv{y}{B}}$,
the inequality $(\normsq{\sdpv{x}{A}} - s) \div \biggap < (\normsq{\sdpv{y}{B}} - s) \div \biggap$ can be  satisfied only if
$(\normsq{\sdpv{y}{B}}-s) \mod \biggap < \smallgap$.
The bad choices for $s$ thus cover at most $(\smallgap / \biggap)$-fraction of the interval $[0, \biggap]$. As $\smallgap / \biggap < K_1/n^4$ (for a suitable constant $K_1$ depending on $|D|$), the probability of a bad choice is at most $K_1/n^4$. There are  $4^{|D|}$ pairs of subsets $A,B \subseteq D$, therefore the probability that the constraint is removed is less than $K_14^{|D|}/n^4$ and so is the expected fraction of removed constraints.

Before analyzing Steps 7 and 8 let us observe that, for any vector $\vect{w}$ such that $1 \geq \norm{\vect{w}} \geq n^{-4r}$,
$$
\pi (\log n) n^{2r} \norm{\vect{w}} \leq \tries{\vect{w}} \leq 2 \pi (\log n) n^{2r} \norm{\vect{w}} + 1.
$$
The first inequality follows from 
$$
\sqrt{(\blockround{\vect{w}}+2)\biggap} = \sqrt{ \biggap ((\normsq{\vect{w}} + 2\biggap - s) \div \biggap) }
\geq $$
$$
\geq \sqrt{ \biggap \frac{\normsq{\vect{w}}+\biggap-s}{\biggap}}
\geq \norm{\vect{w}}
$$
and the second inequality follows from 
$$
\sqrt{(\blockround{\vect{w}}+2)\biggap} \leq \sqrt{ \biggap \frac{(\normsq{\vect{w}} + 2\biggap - s)}{\biggap}} \leq
$$
$$
\leq 
\sqrt{\normsq{\vect{w}} + 2\biggap}
\leq
\sqrt{\normsq{\vect{w}} + 2\normsq{\vect{w}}}
<
2\norm{\vect{w}}.
$$

\textbf{Step 7.} Consider two different subsets $A,B$ of $P_x$ such that $\blockround{\sdpv{x}{A}} = \blockround{\sdpv{x}{B}}$. Suppose that $A \setminus B \neq \emptyset$, the other case is symmetric.
Let $\theta$ be the angle between $\sdpv{x}{A}$ and $\sdpv{x}{B}$. 
As $\sdpv{x}{A}-\sdpv{x}{A \cap B} ( = \sdpv{x}{A \setminus B})$, $\sdpv{x}{B} - \sdpv{x}{A \cap B}$ and $\sdpv{x}{A \cap B}$ are pairwise orthogonal, the angle $\theta$ is greater than or equal to the angle $\theta_A$ between $\sdpv{x}{A}$ and $\sdpv{x}{A \cap B}$. 
(Given three pairwise orthogonal vectors $\sdpv{v}{1}, \sdpv{v}{2}, \sdpv{v}{3}$, the angle between $\sdpv{v}{1}+\sdpv{v}{2}$ and $\sdpv{v}{1}+\sdpv{v}{3}$ is always greater than or equal to the angle between $\sdpv{v}{1}+\sdpv{v}{2}$ and $\sdpv{v}{1}$. 
This is a straightforward calculation using, for instance, dot products. In our situation $\sdpv{v}{1} = \sdpv{x}{A \cap B}$, $\sdpv{v}{2} = \sdpv{x}{A \setminus B}$ and $\sdpv{v}{3} = \sdpv{x}{B \setminus A}$.)
We have $\sin \theta_A = \norm{\sdpv{x}{A \setminus B}} / \norm{\sdpv{x}{A}}$. Since $A \subseteq P_x$, we get $\norm{\sdpv{x}{A \setminus B}} \geq \sqrt{n^{-4r}} = n^{-2r}$
and then
$\sin \theta_A = \norm{\sdpv{x}{A \setminus B}}/ \norm{\sdpv{x}{A}} \geq n^{-2r} / \norm{\sdpv{x}{A}}$, so $\theta \geq \theta_A \geq  n^{-2r} / \norm{\sdpv{x}{A}}$.

The probability that $\vect{q}_i$ does not cut $\sdpv{x}{A}$ and $\sdpv{x}{B}$ is thus at most $1 - n^{-2r} / \pi\norm{\sdpv{x}{A}}$ and the probability that none of the vectors $\vect{q}_1, \dots, \vect{q}_{\tries{\sdpv{x}{A}}}$ cut them is at most
$$
\left( 1 - \frac{n^{-2r}}{\pi\norm{\sdpv{x}{A}}} \right)^{\tries{\sdpv{x}{A}}}
\leq
{\left[\left( 1 - \frac{1}{\pi n^{2r}\norm{\sdpv{x}{A}}} \right)^{\pi n^{2r} \norm{\sdpv{x}{A}}}\right]}^{\log n}
\leq
$$
$$
\leq 
\left(\frac12\right)^{\log n} = \frac1n.
$$
The first inequality uses that $\tries{\sdpv{x}{A}} \geq \pi(\log n) n^{2r} \norm{\sdpv{x}{A}}$ which we observed above. In the second inequality we have used
that $(1 - 1/\eta)^{\eta} \leq 1/2$ whenever $\eta \geq 2$. 

For a single variable there are at most $4^{|D|}$ choices for $A,B \subseteq P_x$, therefore the probability that $x$ is uncut is at most $4^{|D|}/n$. 
The scope of every constraint contains at most $2$ variables, hence the probability that a constraint is removed is at most $2 \cdot 4^{|D|}/n$ and the expected fraction of the constraints removed in this step has the same upper bound.

\textbf{Step 8.} 
Assume that $((x,y),R)$ is a binary constraint and $A \subseteq P_x, B\subseteq P_y$ are such that $\sdpv{x}{A}$ and $\sdpv{y}{B}$ are almost the same.
Let $\theta$ be the angle between $\sdpv{x}{A}$ and $\sdpv{y}{B}$ and $\theta_A$ be the angle between $\sdpv{y}{B}$ and $\sdpv{y}{B}-\sdpv{x}{A}$. 
By the law of sines we have $\norm{\sdpv{x}{A}}/(\sin \theta_A) = \norm{\sdpv{y}{B} - \sdpv{x}{A}}/(\sin \theta)$, and
$$
 \theta \leq 2 \sin \theta =  \frac{2 \norm{\sdpv{y}{B} - \sdpv{x}{A}}}{\norm{\sdpv{x}{A}}} \sin (\theta_A) \leq
 \frac{2\norm{\sdpv{y}{B} - \sdpv{x}{A}}}{\norm{\sdpv{x}{A}}} \leq
 \frac{2 \sqrt{\smallgap}}{\norm{\sdpv{x}{A}}},
$$
where the first inequality follows from $\theta \leq \pi/2$ 
(as the dot product of $\sdpv{x}{A}$ and $\sdpv{y}{B}$ is a sum of nonnegative numbers).
Therefore, the probability that vectors $\sdpv{x}{A}$ and $\sdpv{y}{B}$ are cut by some of the vectors $\vect{q}_i$, $1 \leq i \leq \tries{\sdpv{x}{A}}$ is at most
$$
\tries{\sdpv{x}{A}} \frac{2\sqrt{\smallgap}}{\norm{\sdpv{x}{A}}} \leq
(2 \pi (\log n) n^{2r} \norm{\sdpv{x}{A}} + 1) \frac{2 \sqrt{2|D|^2n^{-4r-4}}}{\norm{\sdpv{x}{A}}} \leq
$$
$$
\leq
K_2(\log n) n^{-2} \leq \frac{K_2}n,
$$
where $K_2$ is a constant.
There are at most $4^{|D|}$ choices for $A,B$, so the probability that our constraint will be removed is less than $K_24^{|D|}/n$.
\end{proof}

\noindent
Now we define the instance $\inst{J}$ and proceed to show that $\inst{J}$ is a weak Prague instance.
Let $\Scopes$ denote the set of pairs which are the scope of some binary constraint of $\inst{I}$ after Step~8, let $V_0$ be the set of variables which are within the scope of some constraint after Step~8, and let $\Scopes^{-1} = \{(x,y): (y,x) \in \Scopes\}$. We put 
\begin{align*}
\inst{J} &= (V_0, D, \{((x,y),P^{\inst{J}}_{x,y}): (x,y) \in \Scopes \cup \Scopes^{-1}\}), \\
P^{\inst{J}}_{x,y} &= \{ (a,b) : \dotprod{\sdpv{x}{a}}{\sdpv{y}{b}} \geq n^{-4r}\}.
\end{align*}

\begin{claim}
The instance $\inst{J}$ is $1$-minimal and $P_x^{\inst{J}} = P_x$.
\end{claim}
\begin{proof}
Let $(x,y) \in \Scopes$ and take an arbitrary constraint $((x,y),R)$ which remained in $\cons{C}$. 

First we prove that $P_{x,y} \subseteq P_x \times P_y$ for every $a,b \in D$. Indeed, if $(a,b) \in P_{x,y}$ then $\dotprod{\sdpv{x}{a}}{\sdpv{y}{b}} \geq n^{-4r}$, therefore $\normsq{\sdpv{x}{a}} = \dotprod{\sdpv{x}{a}}{\sdpv{x}{D}} = \dotprod{\sdpv{x}{a}}{\sdpv{y}{D}} \geq n^{-4r}$, so $a \in P_x$. Similarly, $b \in P_y$.

On the other hand, if $a \in P_x$ then $n^{-4r} \leq \normsq{\sdpv{x}{a}} = \dotprod{\sdpv{x}{a}}{\sdpv{y}{D}}$, thus there exist $b \in D$ such that
$\dotprod{\sdpv{x}{a}}{\sdpv{y}{b}} \geq n^{-4r}/|D| \geq n^{-4r-4}$ (we have used $n^4 \geq |D|$). But then $\dotprod{\sdpv{x}{a}}{\sdpv{y}{b}} \geq n^{-4r}$, otherwise the constraint $((x,y),R)$ would be removed in Step 2. This implies that $(a,b) \in P_{x,y}$. We have shown that the projection of $P_{x,y}$ to the first coordinate contains $P_x$. 
\end{proof}

\noindent
For verification of properties (P2) and (P3) the following observation will be useful.

\begin{claim} \label{cl:walk}
Let $(x,y) \in \Scopes \cup \Scopes^{-1}$, $A \subseteq P_x$, $B = A + (x,y)$. 
If $A = B + (y,x)$, then 
the vectors $\sdpv{x}{A}$ and $\sdpv{y}{B}$ are almost the same.
In the other case, i.e. if $A \varsubsetneq B + (y,x)$, then 
$\blockround{\sdpv{y}{B}} > \blockround{\sdpv{x}{A}}$.
\end{claim} 
\begin{proof}
The number $\normsq{\sdpv{y}{B}-\sdpv{x}{A}}$ is equal to
$$
\dotprod{\sdpv{y}{B}}{\sdpv{y}{B}} - \dotprod{\sdpv{x}{A}}{\sdpv{y}{B}} 
- \dotprod{\sdpv{x}{A}}{\sdpv{y}{B}} + \dotprod{\sdpv{x}{A}}{\sdpv{x}{A}}
=
$$
$$
=
\dotprod{\sdpv{x}{D}}{\sdpv{y}{B}} - \dotprod{\sdpv{x}{A}}{\sdpv{y}{B}} 
- \dotprod{\sdpv{x}{A}}{\sdpv{y}{B}} + \dotprod{\sdpv{x}{A}}{\sdpv{y}{D}}
=
\dotprod{\sdpv{x}{D \setminus A}}{\sdpv{y}{B}}
+ \dotprod{\sdpv{x}{A}}{\sdpv{y}{D \setminus B}}.
$$
No pair $(a,b)$, with $a\in A$ and $b\in D\setminus B$, is in $P^{\inst{J}}_{x,y}$ so
the dot product $\dotprod{\sdpv{x}{a}}{\sdpv{y}{b}}$ is smaller than $n^{-4r}$. 
Then in fact $\dotprod{\sdpv{x}{a}}{\sdpv{y}{b}} < n^{-4r-4}$ otherwise all the constraints with scope $(x,y)$ would be removed in Step 2.
It follows that 
the second summand is always at most $|D|^2 n^{-4r-4}$ and the first summand has the same upper bound in the case $B+(y,x) = A$.

Moreover, $\normsq{\sdpv{y}{B}}-\normsq{\sdpv{x}{A}}$ is equal to
$$
\dotprod{\sdpv{y}{B}}{\sdpv{y}{B}} - \dotprod{\sdpv{x}{A}}{\sdpv{x}{A}}
=
\dotprod{\sdpv{x}{D}}{\sdpv{y}{B}} - \dotprod{\sdpv{x}{A}}{\sdpv{y}{D}}
=
$$
$$
=
\dotprod{\sdpv{x}{D}}{\sdpv{y}{B}} - \dotprod{\sdpv{x}{A}}{\sdpv{y}{B}} - \dotprod{\sdpv{x}{A}}{\sdpv{y}{D \setminus B}}
=
\dotprod{\sdpv{x}{D \setminus A}}{\sdpv{y}{B}} - \dotprod{\sdpv{x}{A}}{\sdpv{y}{D \setminus B}}.
$$
If $B + (y,x) = A$ then we have a difference of two nonnegative numbers each less than or equal $|D|^2 n^{-4r-4}$, therefore the absolute value of this expression is at most $\smallgap$. But then $\blockround{\sdpv{x}{A}} = \blockround{\sdpv{y}{B}}$, otherwise all constraint with scope $(x,y)$ or $(y,x)$ would be removed in Step 5. Using the previous paragraph, it follows that $\sdpv{x}{A}$ and $\sdpv{y}{B}$ are almost the same.

If $B + (y,x)$ properly contains $A$ then the first summand $\dotprod{\sdpv{x}{D\setminus A}}{\sdpv{y}{B}}$ is greater than or equal to $n^{-4r}$, so the whole expression is at least $n^{-4r} - |D|^2 n^{-4r-4} > \biggap$ and thus $\blockround{\sdpv{y}{B}} > \blockround{\sdpv{x}{A}}$.
\end{proof}

\begin{claim} 
$\inst{J}$ is a weak Prague instance.
\end{claim}
\begin{proof}
\textbf{(P2).}
Let $A \subseteq P_x$ and let $p = (x_1, \dots, x_i)$ be a pattern in $\inst{J}$ from $x$ to $x$ (i.e. $x_1 = x_i = x$). 
By the previous claim 
$\blockround{\sdpv{x}{A}} = \blockround{(\vect{x}_i)_{A+(x_1, \dots, x_i)}} \geq
\blockround{(\vect{x}_{i-1})_{A+(x_1, \dots, x_{i-1})}}$ $\geq \dots \geq \blockround{(\vect{x}_2)_{A+(x_1,x_2)}} \geq \blockround{\sdpv{x}{A}}$.
It follows that all these inequalities must in fact be equalities and, by applying the claim again,
we get that the vectors $(\vect{x}_j)_{A+(x_1,x_2, \dots, x_j)}$ and $(\vect{x}_{j+1})_{A+(x_1, x_2, \dots, x_{j+1})}$ are almost the same 
and, moreover, $A+(x_1,x_2, \dots, x_{j+1}) + (x_{j+1}, x_j) = A+(x_1, x_2, \dots, x_j)$
for every $1 \leq j < i$.
Therefore $A + p - p = A$ as required. 

\textbf{(P3).}
Let $A \subseteq P_x$, let $p_1 = (x_1, \dots, x_i), p_2$ be two patterns from $x$ to $x$ such that $A + p_1 + p_2 = A$ and let $B = A+p_1$. For contradiction assume $A \neq B$. 
The same argument as above proves that the vectors $(\vect{x}_j)_{A+(x_1,x_2, \dots, x_j)}$ and $(\vect{x}_{j+1})_{A+(x_1, x_2, \dots, x_{j+1})}$ are almost the same for every $1 \leq j < i$, and then $\blockround{\sdpv{x}{A}} = \blockround{\sdpv{x}{B}}$. There exists $k \leq \tries{\sdpv{x}{A}}$ such that
$\sgn \dotprod{\sdpv{x}{A}}{\vect{q}_k} \neq \sgn \dotprod{\sdpv{x}{B}}{\vect{q}_k}$, otherwise $x$ is uncut and all constraints whose scope contains $x$ would be removed in Step 7. 
But this leads to a contradiction, since $\sgn \dotprod{(\vect{x}_{j})_{A+(x_1, \dots, x_{j})}}{\vect{q}_k} = \sgn \dotprod{(\vect{x}_{j+1})_{A+(x_1, \dots, x_{j+1})}}{\vect{q}_k}$ for all $1 \leq j < i$, otherwise the constraints with scope $(x_{j},x_{j+1})$ would be removed in Step 8.
\end{proof}

\noindent
Observe that for every unary constraint $(x,R)$ we have $P_x \subseteq R$ (from Step 3) and for every binary constraint $((x,y),R)$ we have $P_{x,y} \subseteq R$. Since we have removed at most $(K/n)$-fraction of the constraints from $\cons{C}$, 
a potential solution to $\inst{J}$ is an assignment for the original instance $\inst{I}$ of value at least $1-K/n$. Also, the instance $\inst{J}$ is nontrivial because, for each $x \in V$, there exists at least one $a \in D$ with $\normsq{\sdpv{x}{a}} > 1/n^4$ (recall that we assume $n > |D|$).

The only problem is that the CSP over the constraint language of $\inst{J}$ (consisting of $P^{\inst{J}}_{x,y}$'s) does not necessarily have bounded width. This is why we form the algebraic closure $\inst{J}'$ of $\inst{J}$:
\begin{align*}
\inst{J'} &= (V_0,D, \{((x,y),P^{\inst{J}'}_{x,y}): (x,y) \in \Scopes \cup \Scopes^{-1}\}), \\
P_{x,y}^{\inst{J}'} &= \{(f(a_1, a_2, \dots), f(b_1, b_2, \dots)): f \in \Pol(\Gamma), \\
  & \quad \quad \quad (a_1,b_1), (a_2,b_2), \dots \in P^{\inst{J}}_{x,y}\}
\end{align*}
The new instance still has the property that $P_x^{\inst{J}'}$ (which is equal to $\{f(a_1, a_2, \dots): f \in \Pol(\Gamma), a_1, a_2, \dots \in P_x\}$) is a subset of $R$ for every unary constraint $(x,R)$, and $P_{x,y}^{\inst{J}'} \subseteq R$ for every binary constraint $((x,y),R)$, since the constraint relations are preserved by every polymorphism of $\Gamma$. Moreover, every polymorphism of $\Gamma$ is a polymorphism of the constraint language $\Lambda'$ of $\inst{J}'$, therefore $\CSP(\Lambda')$ has bounded width (see, for instance, Theorem~\ref{thm:bw_char}; technically, $\Lambda'$ does not need to be a core, but we can simply add to $\Lambda'$ all the singleton unary relations).

By Proposition~\ref{prop:closure}, $\inst{J}'$ is a weak Prague instance.
Therefore $\inst{J}'$ (and thus $\inst{I}$ after Step~8) has a solution by Theorem~\ref{thm:bw}.
Clearly steps 1 to 8 can be done in time polynomial with respect to $m$ and $n^n$, 
the closure required by Proposition~{\ref{prop:closure}} is computed in time linear in $m$
and then a solution to $\inst{J}'$ can be found in polynomial time by Theorem~{\ref{thm:search_problem}}.
This concludes the proof.

\subsection{Derandomization}

The following theorem is a deterministic version of Theorem~\ref{THM:CORE}. The statement is almost the same except the running time is polynomial in  $2^{n^2\log^2 n}$ instead of $n^n$. 

\begin{theorem} \label{THM:COREderandom}
Let $\Gamma$ be a core constraint language over $D$ containing at most binary relations. 
If $\CSP(\Gamma)$ has bounded width, then there exists
a deterministic algorithm which given an instance $\inst{I}$ of $\CSP(\Gamma)$ and an output of the basic SDP relaxation with value at least $1 - 1/n^{4n}$ (where $n$ is a natural number) produces an assignment with value at least $1 - K/n$, where $K$ is a constant depending on $|D|$. The running time is polynomial in $m$ (the number of constraints) and $2^{n^2\log^2 n}$.
\end{theorem}

\begin{proof}
The algorithm is the same as in the proof of Theorem~\ref{THM:CORE} except we need to avoid the random choices in Steps 1, 4 and 6.

The random choices in Step 1 and Step 4 can be easily avoided. In Step 1 we can try all $(n-1)$ possible choices for $r$ and
in Step 4 we can try all choices for $s$ from some sufficiently dense finite set,  for instance $\{0, \biggap/n^4, 2 \biggap/n^4, \dots\}$. The only difference is that bad choices for $s$ can cover a slightly bigger part of the discrete set than $\smallgap/\biggap$ (namely $(\biggap/n^4+\smallgap)/\biggap$) and we get a slightly worse constant $K_1$.

For the derandomization of Step 6 we first increase the constant in the definition of $t(\vect{w})$, say
$
t(\vect{w}) = \roundup{4 (\log n) \dots}.
$
Next we use Theorem 1.3. in \cite{KRS11} from which it follows that we can find (in polynomial time with respect to $|Q|$)  a set $Q$ of unit vectors such that
$$|Q| = (|V||D|)^{1 + o(1)}2^{O(\log^2 (1/\kappa))}$$
and such that, for any vectors $\vect{v}$, $\vect{w}$ with angle $\theta$ between them, the probability that a randomly chosen vector from $Q$ cuts $\vect{v}$ and $\vect{w}$ differs from $\theta/\pi$ by at most $\kappa$.  We choose $\kappa = 1/n^{2n} = 1/2^{2n \log n}$, therefore
$$
|Q| \leq K_5 m^{K_6} \ (2^{n^2 \log^2 n})^{K_7}
$$
where we have used $|V|=O(m)$ which is true whenever every variable is in the scope of some constraint (we can clearly make this assumption without loss of generality).

Now if we choose $\vect{q}_1, \vect{q}_2, \dots, \vect{q}_{\roundup{4(\log n)n^{2n}}}$ from $Q$ independently uniformly at random, the estimates derived in Steps 7 and 8 remain almost unchanged: The probability that $\vect{q}_i$ does not cut $\vect{x}_A$ and $\vect{x}_B$ in Step 7 is at most
$1 - n^{-2r}/\pi\norm{\sdpv{x}{A}} + \kappa \leq 1 - n^{-2r}/4\norm{\sdpv{x}{A}}$ (for a sufficiently large $n$), and
the probability in Step 8 that vectors $\sdpv{x}{A}$ and $\sdpv{y}{B}$ are cut by $\vect{q}_i$ is at most $2\sqrt{\smallgap}/\norm{\sdpv{x}{A}} + \kappa \leq 4\sqrt{\smallgap}/\norm{\sdpv{x}{A}}$.

Unfortunately we cannot try all possible ${\roundup{4(\log n)n^{2n}}}$-tuples of vectors from $Q$ as there are too many of them. 
To resolve this problem we apply the method of conditional expectations with a pessimistic estimate. 
We choose the vectors one by one keeping the sum of two estimates, $\estimcut{i}$ and $\estimuncut{i}$, reasonably small, where $\estimcut{i}$ is an estimate of the number of uncut pairs of vectors $\sdpv{x}{A},\sdpv{x}{B}$ in the same layer (i.e. $\tries{\vect{x}_A} = \tries{\vect{x}_B}$) after we already chose vectors $\vect{q}_1, \dots, \vect{q}_i$ (important in Step 7) and $\estimuncut{i}$ is an estimate of the number of cut pairs of almost the same vectors $\vc{x}_A, \vc{y}_B$ (important in Step 8).

To define the pessimistic estimates, we create two lists of pairs of vectors. The list $\listcut$ contains the pairs we want to cut because of Step 7:
For every constraint $C$ we include to $\listcut$ the pairs $(\sdpv{x}{A}, \sdpv{x}{B})$ such that $x$ is in the scope of $C$, $A \neq B \subseteq P_x$, and $\blockround{\sdpv{x}{A}} = \blockround{\sdpv{x}{B}}$. One pair $(\sdpv{x}{A},\sdpv{x}{B})$ can appear more than once in the list --- the number of occurrences is equal to the number of constraints whose scope contains $x$. 
The other list, $\listuncut$, consists of those pairs which we do not want to cut because of Step 8. For each constraint $C = ((x,y),R)$ we add to $\listuncut$ the pairs $(\sdpv{x}{A},\sdpv{y}{B})$ such that $A,B \subseteq P_x$ and $\sdpv{x}{A}$ is almost the same as $\sdpv{y}{B}$. We denote by $\probcut{\sdpv{x}{A}}{\sdpv{x}{B}}$ the upper bound derived above  for the probability that a random vector from $Q$ does not cut $\sdpv{x}{A}$ and $\sdpv{x}{B}$ and by $\probuncut{\sdpv{x}{A}}{\sdpv{y}{B}}$ the upper bound for the probability that a random vector from $Q$ cuts $\sdpv{x}{A}$ and $\sdpv{y}{B}$, i.e.
$$
\probcut{\sdpv{x}{A}}{\sdpv{x}{B}} = 1 - \frac{n^{-2r}}{4 \norm{\sdpv{x}{A}}}, \quad
\probuncut{\sdpv{x}{A}}{\sdpv{y}{B}} = 4\sqrt{\smallgap}/\norm{\sdpv{x}{A}}\enspace.
$$

Suppose we have already selected vectors $\vect{q}_1, \dots, \vect{q}_i$, $0 \leq i$. Let us denote by $\listi{i}$ the sublist of $\listcut$ formed by pairs which are not cut by the vectors $\vect{q}_1, \dots, \vect{q}_i$ and by $\cuti{i}$ the number of pairs from $\listuncut$ cut by these vectors. If we now choose vectors $\vect{q}_{i+1}, \vect{q}_{i+2}, \dots$ from $Q$  independently uniformly at random we can give an upper estimate for the expected fraction of pairs from $\listcut$ which will remain uncut 
$$
\estimcut{i} = \frac{1}{|\listcut|} \sum_{(\sdpv{x}{A},\sdpv{x}{B}) \in \listi{i}} \probcut{\sdpv{x}{A}}{\sdpv{x}{B}}^{\max(\tries{\sdpv{x}{A}}-i,0)}\enspace,
$$
and for the expected fraction of pairs from $\listuncut$ which are cut
$$
\estimuncut{i} = \frac{1}{|\listuncut|}\left(\cuti{i} + \sum_{(\sdpv{x}{A},\sdpv{y}{B}) \in \listuncut} \max(\tries{\sdpv{x}{A}}-i,0) \probuncut{\sdpv{x}{A}}{\sdpv{y}{B}}\right)\enspace .
$$
Calculations made in the proof of Theorem~\ref{THM:CORE} show that
$$
\estimcut{0} = O\left(\frac1n\right), \quad
\estimuncut{0} = O\left(\frac1n\right)\enspace. 
$$

The estimates are becoming less pessimistic as $i$ increases. More precisely, given vectors $\vect{q}_1, \dots, \vect{q}_i$, if we choose $\vect{q}_{i+1}$ uniformly at random from $Q$, the expected value of $\estimcut{i+1} + \estimuncut{i+1}$  is at most $\estimcut{i} + \estimuncut{i}$. Indeed, if $(\sdpv{x}{A},\sdpv{x}{B})$ is in $\listi{i}$ and $\tries{\sdpv{x}{A}}\leq i$ then 
the contribution of this pair to both sums $\estimcut{i}$ is one and to $ \estimcut{i+1}$ is zero or one. 
If, on the other hand, $\tries{\sdpv{x}{A}} > i$ then the contribution of this pair to the sum in $\estimcut{i}$ is $\probcut{\sdpv{x}{A}}{\sdpv{x}{B}}^{\tries{\sdpv{x}{A}}-i}$ and the contribution to the sum in $\estimcut{i+1}$ is either zero, when the pair is cut by $\vect{q}_{i+1}$, or is equal to $\probcut{\sdpv{x}{A}}{\sdpv{x}{B}}^{\tries{\sdpv{x}{A}}-i-1}$, when the pair is not cut by $\vect{q}_{i+1}$. The latter option happens with probability at most $\probcut{\sdpv{x}{A}}{\sdpv{x}{B}}$. Expected contribution of the pair to the sum in $\estimcut{i+1}$ is therefore
 at most $\probcut{\sdpv{x}{A}}{\sdpv{x}{B}} \probcut{\sdpv{x}{A}}{\sdpv{x}{B}}^{\tries{\sdpv{x}{A}}-i-1} = 
 \probcut{\sdpv{x}{A}}{\sdpv{x}{B}}^{\tries{\sdpv{x}{A}}-i}$ which is the same as the contribution of this pair to the sum  in $\estimcut{i}$.
 We conclude that the expected value of $\estimcut{i+1}$ is less than or equal to $\estimcut{i}$. Similarly, the expected value of $\estimuncut{i+1}$ is at most $\estimuncut{i}$. It follows that the expected value of $\estimcut{i+1} + \estimuncut{i+1}$ is less than or equal to $\estimcut{i} + \estimuncut{i}$ as claimed.
 
This leads to the following (deterministic) algorithm. For $i = 0,1, \dots, \roundup{4(\log n)n^{2n}}-1$ we select any $\vect{q}_{i+1} \in Q$ such that $\estimcut{i+1} + \estimuncut{i+1} \leq \estimcut{i} + \estimuncut{i}$. It remains to observe that this choice of vectors $\vect{q}_1, \vect{q}_2, \dots$ gives us a sufficient upper bound on the fraction of constraints removed in Steps 7 and 8. We denote by $\wrongcut$ the fraction of pairs in $\listcut$ which are not cut by the selected vectors and $\wronguncut$ the fraction of pairs in $\listuncut$ which are cut by these vectors. By construction, $\wrongcut + \wronguncut \leq \estimcut{0} + \estimuncut{0}$. In Step 7, a constraint $C$ is removed if some pair $(\sdpv{x}{A},\sdpv{x}{B})$ from $\listcut$, where $x$ is in the scope of $C$, is not cut. The number of such pairs is at most $2 \cdot 4^{|D|}$. It follows that the fraction of constraints removed in this step is at most $2\cdot 4^{|D|} \wrongcut$. Similarly, in Step 8 we remove at most $4^{|D|} \wronguncut < 2 \cdot 4^{|D|} \wronguncut$ fraction of the constraints. Altogether, we remove at most the following fraction of constraints.
$$
2 \cdot 4^{|D|} (\wrongcut + \wronguncut) \leq 2 \cdot 4^{|D|} \left(\estimcut{0} + \estimuncut{0}\right)
= O\left(\frac1n\right)
$$ 
\end{proof}

Finally, we prove the 
deterministic version of the main theorem.

\begin{theorem} \label{THM:MAINderandom}
If $\CSP(\Gamma)$ has bounded width then it is robustly solvable. 
The algorithm  returns an assignment satisfying $(1 - O(\log \log (1/\varepsilon)/{\sqrt{\log (1/\varepsilon)}}))$-fraction of the constraints given a $(1-\varepsilon)$-satisfiable instance.
\end{theorem}

\begin{proof}
The proof is almost the same as for Theorem~\ref{THM:MAIN}  except we need to ensure that $2^{n^2 \log^2 n}$ is polynomial in $m$. Therefore we need to choose a smaller value for $n$, say
$$
n = \rounddown{\frac{\sqrt{\log \omega}}{4\log \log \omega}}\enspace.
$$
The inequality $v \geq 1 - 1/n^{4n}$ still holds since the expression on the right hand side is even smaller than in Theorem~\ref{THM:MAIN}.
The algorithm runs in polynomial time as
$$
2^{n^2\log^2 n}  
\leq 
2^{\left(\frac{\sqrt{\log \omega}}{4\log \log \omega}\right)^2 \log^2 \left(\frac{\sqrt{\log \omega}}{4\log \log \omega}\right)}
\leq 
2^{\left(\frac{\sqrt{\log \omega}}{4\log \log \omega}\right)^2 \log^2 \left(\sqrt{\log \omega}\right)}
\leq
2^{\frac1{4^3} \log \omega}
\leq
\omega
\leq
m
$$
and the fraction of satisfied constraint is at least $1 - K/n$, where
\begin{align*}
\frac{n}{K} &\geq \frac1K \left( \frac{\sqrt{\log \omega}}{4\log \log \omega} - 1 \right)
\geq K_3 \frac{\sqrt{\log (1/2\varepsilon)}}{\log \log (1/2\varepsilon)} \geq \\
&\geq K_4 \frac{\sqrt{\log (1/\varepsilon)}}{\log \log (1/\varepsilon)}
\enspace,
\end{align*}
therefore the fraction of satisfied constraints is at least 
$$
1 - O\left(\frac{\log{\log(1/\varepsilon)}}{\sqrt{\log (1/\varepsilon)}}\right)\enspace.
$$
\end{proof}

\subsection{Final remarks}

The quantitative dependence of $g$ on $\varepsilon$ is not very far from the (UGC-) optimal bound for \HORN. Is it possible to get rid of the extra $\log \log (1/\varepsilon)$? The question of the optimal quantitative dependence of $g$ on $\varepsilon$ is discussed in more detail in \cite{DK}.

The presented straightforward derandomization using a result from \cite{KRS11} has $g(\varepsilon) = O(\log \log (1/\varepsilon)/{\sqrt{\log (1/\varepsilon)}})$. How to improve it to match the randomized version?

\bibliographystyle{abbrv}

\bibliography{CSPbibRobust}

\end{document}